\documentclass[final,3p,times]{elsarticle}
\usepackage{amsmath}
\usepackage{graphicx}
\usepackage{enumitem}
\usepackage{natbib}
\usepackage{url} % not crucial - just used below for the URL 
\usepackage{amssymb}
\usepackage{amsthm}
\usepackage{amsmath}
\usepackage{booktabs}
\usepackage{etoolbox}
\usepackage{algorithm}
\usepackage{bm}
\usepackage[table]{xcolor}

\newtheorem{prop}{Proposition}% 
\newtheorem{lem}{Lemma}% 
\newtheorem{defn}{Definition}%
\newtheorem{rem}{Remark}%

\journal{t.b.d.}

\begin{document}

\begin{frontmatter}

\title{Principal Component Copulas for Capital Modeling and Systemic Risk}

\author[UvT,Achmea]{K.B. Gubbels}
\ead{k.b.gubbels@tilburguniversity.edu}

\affiliation[UvT]{organization={Department of Econometrics and Operations Research, Tilburg University},%Department and Organization
            city={Tilburg},
            country={The Netherlands}}

 \affiliation[Achmea]{organization={Achmea},%Department and Organization
            city={Zeist},
            country={The Netherlands}}       
            
\author[Achmea]{J.Y. Ypma}

\author[UU]{C.W. Oosterlee}

\affiliation[UU]{organization={Mathematical Institute, Utrecht University},%Department and Organization
            city={Utrecht},
            country={The Netherlands}}

\begin{abstract}

We introduce a class of copulas that we call Principal Component Copulas (PCCs). This class combines the strong points of copula-based techniques with principal component analysis (PCA), which results in flexibility when modelling tail dependence along the most important directions in high-dimensional data. We obtain theoretical results for PCCs that are important for practical applications. In particular, we derive tractable expressions for the high-dimensional copula density, which can be represented in terms of characteristic functions. We also develop algorithms to perform Maximum Likelihood and Generalized Method of Moment estimation in high-dimensions and show very good performance in simulation experiments. Finally, we apply the copula to the international stock market to study systemic risk. We find that PCCs lead to excellent performance on measures of systemic risk due to their ability to distinguish between parallel and orthogonal movements in the global market, which have a different impact on systemic risk and diversification. As a result, we consider the PCC promising for capital models, which financial institutions use to protect themselves against systemic risk.

\end{abstract}

\begin{keyword}
Copulas
\sep Principal Component Analysis
\sep Dependence modelling
\sep Capital modelling
\sep Systemic risk

%% PACS codes here, in the form: \PACS code \sep code

%% MSC codes here, in the form: \MSC code \sep code
%% or \MSC[2008] code \sep code (2000 is the default)

\end{keyword}

\end{frontmatter}

\section{Introduction}\label{secIntro}

Modelling the dependence structure of a multivariate distribution plays a crucial role in finance and risk management. Copulas are a popular approach, since they allow to separate the choices for the marginal distributions from the dependence structure \citep{Nelson2006, Joe2014}. The Gaussian copula has popularized the use of copulas in finance due to its simplicity, especially in high-dimensional applications \citep{Li2000}. However, its disadvantage is that it cannot capture tail dependence. As a consequence, the Student {\it t} copula and its extensions have become popular \citep{Demarta2005}. Although these copulas are often used in practice, also other copula methods have become available that enhance flexibility. Vine copulas are constructed by sequentially applying bivariate copulas to build higher-dimensional copulas \citep{Aas2009}. The resulting pair copula constructions are flexible, but model complexity can quickly increase, which can be a disadvantage for internal capital models of financial institutions. 

More recently, factor copulas have received a steady increase of interest in the literature \citep{Krupskii2013, Creal2015, Oh2017, Opschoor2020, Duan2022}. Factor copulas for stock returns have shown good  performance in capturing systemic risk \cite{Oh2017} and in Value-at-Risk forecasts \cite{Duan2022}. As a result, such copulas are promising for capital modelling. Although the attention for factor copulas has been increasing, there also remain challenges. In general, no analytic expressions exist for the copula density, which makes estimation nontrivial. \cite{Oh2017} have developed a simulation-based estimation technique, but this method is numerically challenging, which has limited the number of estimated parameters. Another important challenge for factor copulas is to determine the appropriate factor structure. For return data, the factors are often based on sectoral codes \citep{Creal2015, Oh2017, Opschoor2020}. However, this approach is suboptimal compared to a data-driven clustering approach \citep{Oh2023}. This shows that finding the appropriate factor structure is in general a non-trivial task.  

In this article, we introduce a novel class of copulas to overcome the identified challenges. We develop a procedure to generate high-dimensional implicit copula structures using the convolution of stochastic generators with analytically tractable characteristic functions. Because these stochastic generators are uncorrelated, we identify them as Principal Components (PCs). Principal Component Analysis (PCA) is convenient in automatically detecting the main drivers of high-dimensional dependence. It is the most-used dimension reduction technique due to its user friendliness and wide applicability. In finance, PCA can be applied to estimate common factors in asset returns \citep{Stock2002}. However, a systematic study of the implicit copula that is generated by PCA has so far been lacking. We call such an implicit copula a PCC, and we obtain new theoretical results for PCCs that we consider promising for practical applications. First of all, we derive the copula density for a PCC, which is represented in terms of characteristic functions. This allows us to determine the joint copula density in high dimensions, while requiring only one-dimensional integrals. Second, we derive an analytic expression for the tail dependence of a copula that is generated with hyperbolic and normal distributions, which have exponential and Gaussian tails. This leads to a novel example of a tail-dependent copula, where previous analytic results have focused on heavy-tailed factors \citep{Hult2002,Malevergne2004,Oh2017}. 

Another contribution of our article is that we develop efficient algorithms to estimate high-dimensional PCCs. We propose the Generalized Hyperbolic (GH) distribution as stochastic generator family, because of its flexibility and its analytic tractability. In order to achieve Maximum Likelihood Estimation (MLE), we develop numerical methods to accurately determine the high-dimensional copula density at limited computational costs. We propose the Fast Fourier Transform \citep{Carr1999} and the COS method \citep{Fang2009} to efficiently compute the required one-dimensional integrals over characteristic functions. Due to the fast evaluation of the copula density, we achieve efficient MLE for flexible copula structures. In high dimensions, the number of correlation parameters grows rapidly. Full MLE remains feasible only when additional structure is applied to the correlation matrix, so that it can be parameterized in a parsimonious way. Although such restricted models can be beneficial, it is usually not desirable to make strong a priori assumptions on the correlation matrix. Therefore, we also propose a second estimation scheme. For high-dimensional copulas, it is natural to make a distinction between correlation parameters (second-moment parameters) and shape parameters (higher-moment parameters). We develop a novel algorithm that is designed to efficiently estimate a large number of correlation parameters using the method of moments and a small number of shape parameters based on MLE. We show excellent performance of this algorithm in high-dimensional simulation experiments.

Our final contribution is that we apply the developed copula to the international stock market in order to study global diversification and systemic risk. Recently, several copula studies have appeared to analyze systemic risk in bond markets \citep{Lucas2014}, stock markets \citep{Oh2017}, credit derivative markets \citep{Oh2018} and commodity markets \citep{Ouyang2022}. We expand upon existing studies by introducing the market distress cube to empirically study systemic risk in three directions, namely more extreme crashes, a larger global market size, and an increasing market share. We also develop a binomial test to assess the ability of various copula specifications to capture high-dimensional dependence. We find that PCCs outperform existing copulas based on information criteria and measures of systemic risk. The enhanced performance arises from their ability to distinguish between parallel movements and orthogonal movements in the global market, which have different consequences for systemic risk and diversification. As a result, we consider the PCC promising for quantitative risk management and capital modeling, which financial institutions use to protect themselves against systemic risk.

The remainder of the article is structured as follows. Section \ref{secPcc} presents the class of Principal Component Copulas. Section \ref{secEst} describes various estimation techniques. Section \ref{secApplications} presents a simulation study that focuses on the performance of the estimators, and an empirical study of tail risk in weekly returns of various market indices over the period 1998\textendash2022. 

\section{Principal Component Copulas}\label{secPcc}

The main advantage of a copula is the possibility to separate the modelling of marginal distributions from the modelling of the dependence structure \citep{Nelson2006, Joe2014}. In high dimensions, implicit copulas are commonly applied, which are constructed from analytically tractable multivariate distributions. Well-known examples of implicit copulas are Gaussian copulas, normal mixture copulas such as the (skew) $t$ copula, and factor copulas. A recent review of implicit copulas is given by \cite{Smith2023}. In general, the joint distribution function $F_{\bm{Y}}(y_1,...,y_d)$ for the stochastic vector $\mathbf{Y}$ can be expressed as $F_{\mathbf{Y}}(y_1,...,y_d)=C_{\mathbf{Y}}(F_{Y_1}(y_1),...,F_{Y_d}(y_d))$ with $F_{Y_i}$ the marginal distribution function for index $i$ and $C_{\mathbf{Y}}(u_1,...,u_d)$ the implicit copula belonging to $F_{\bm{Y}}$. If the marginal distributions are continuous, the implicit copula is unique, which is the setting of relevance to us. The copula density is then given by:
\begin{eqnarray}
c_{\bm{Y}}(u_{1},...,u_{d})&=&\frac{f_{\bm{Y}}(F_{Y_1}^{-1}(u_1),...,F_{Y_d}^{-1}(u_d))}{f_{Y_1}(F_{Y_1}^{-1}(u_1))...f_{Y_d}(F_{Y_d}^{-1}(u_d))}.\label{eqCopDens}
\end{eqnarray}
In this article, we develop a flexible and tractable procedure to develop novel $d$-dimensional implicit copula structures for the stochastic vector $\bm{Y}$ using convolutions of stochastic variables with analytically tractable probability densities and characteristic functions.

\subsection{Data generating process for the implicit copula}\label{subSecCop}

We start with considering a collection of $d$ stochastic variables $P_j$, which are called the generators of the copula. We consider convolutions of these generators using linear transformation matrices $\bm{W}\in \mathbb{R}^{d\times d}$, so that $\bm{Y}=\bm{WP}$. In first instance, the matrix $\bm{W}$ is considered to have full rank $d$, so that $\bm{W}$ is invertible. We comment on the dimension-reduced case later on. In component form, the convolution of stochastic generators $P_j$ to generate $Y_i$ is given by:
\begin{equation}
 Y_i =\sum_{j=1}^d  w_{i,j}P_j  \label{eqPca}
\end{equation}

Without loss of generality for the copula and to avoid overspecification, we take the variables $Y_i$ to be standardized with $E[Y_i]=0$ and $\text{var}[Y_i]=1$. As a result, the first and second moments of $\bm{Y}$ are fixed by the linear correlation matrix $\bm{\rho}_{\bm{Y}}$. Using the spectral theorem, we obtain the decomposition $\bm{\rho}_{\bm{Y}}=\bm{W} \bm{\Lambda} \bm{W}'$ with $\bm{W}$ the orthogonal matrix of eigenvectors and $\bm{\Lambda}$ the diagonal matrix of nonnegative eigenvalues $\Lambda_{j}$. Consequently, we identify $\bm{\Lambda}$ with the covariance matrix of the generators $\bm{P}$. This means that the implicit copula generators are uncorrelated and that generator $P_j$ has variance $\Lambda_{j}$. To each generator $P_j$ belongs an eigenvector $\bm{w}_{j}$, which is the $j$-th column of $\bm{W}$. Because the generators are uncorrelated and are connected to the eigenvectors of the copula correlation matrix, we identify them as Principal Components (PCs). The resulting copula is therefore called a PCC. We do not require the PCs to be independent. Because $\bm{W}$ is orthogonal, we have that $\bm{P}=\bm{W}'\bm{Y}$. The expectation value of $P_j$ is given by $E[{P_j}]=0$. The generators are ranked based on their eigenvalues, so that $P_1$ has the highest variance $\Lambda_1$. The first eigenvector $\bm{w}_{1}$ is usually a parallel vector for financial data with the same sign for all indices $w_{i,1}$. In that case, we see from Eq.  (\ref{eqPca}) that $P_1$ generates a joint parallel movement of the indices $Y_i$.   

The copula construction therefore starts from uncorrelated generators or PCs. The correlation structure of the implicit copula is consequently determined by the diagonal covariance matrix $\bm{\Lambda}$ for the generators and by the orthogonal transformation matrix $\bm{W}$. Flexible tail dependence structures for the copula can be obtained using appropriate families of generator distributions. If we specify the Data Generating Process (DGP) for the stochastic vector ${\bm P}$ by its continuous multivariate distribution function $F_{\bm P}$, then the distribution function $F_{\bm Y}$ and the corresponding copula $C_{\bm Y}$ are obtained by the linear transformation of Eq. (\ref{eqPca}). We therefore arrive at the following definition.

\begin{defn}\label{defPCC}
Consider two $d$-dimensional stochastic vectors $\bm{Y}$ and $\bm{P}$ that are related to each other through the linear transformation $\bm{Y}=\bm{WP}$ with $\bm{W}$ an orthogonal matrix of eigenvectors of the linear correlation matrix $\bm{\rho}_{\bm{Y}}=\bm{W} \bm{\Lambda} \bm{W}'$. Let $F_{\bm P}$ be the multivariate distribution function of the uncorrelated stochastic vector ${\bm P}$, which has diagonal covariance matrix $\bm{\Lambda}$. A Principal Component Copula $C_{\bm{Y}}$ is defined as the implicit copula for $\bm{Y}$ that follows from $F_{\bm P}$ and the linear transformation $\bm{Y}=\bm{WP}$.
\end{defn}

The core idea of this definition and the related copula construction is that novel multivariate distribution functions of the stochastic vector $\bm{Y}$ can be constructed from convolutions of the stochastic vector $\bm{P}$. This means that we apply the techniques of stochastic convolution to derive (semi-)analytic results for the multivariate distribution of $\bm{Y}$, based on suitable choices for the stochastic generator $\bm{P}$. Although this main idea is general, we introduce two specific examples that we use in our empirical studies.

\subsection{Examples}\label{subsecEx}

In high-dimensional financial applications, the eigenvalue spectrum of correlation matrices typically consists of a (few) large distinct eigenvalue(s) and a dense set of smaller eigenvalues \citep{Bun2017}. The component(s) with highest eigenvalue(s) determine the dominant collective movement(s) of the market, where the first PC typically describes the parallel movement of all indices. This first generator is therefore important for capturing systemic risk and for capital modelling, since during a market crash the indices have the tendency to go down jointly. The smallest eigenvalues typically do not differ from the spectrum that describes random noise \citep{Bun2017}. Therefore, it is natural to consider PCCs that assign a more detailed model structure to the first PC(s), while the higher PCs are modelled with less detail. When the first generator of the implicit copula is skewed and leptokurtic, high-dimensional asymmetric tail dependence can be generated in a parsimonious way. The examples that we introduce next are based on these properties and are relevant for our case study of the dependence structure in global financial markets.

\subsubsection{Hyperbolic-Normal PCC}\label{secHypNorPcc}

We introduce the Hyperbolic-Normal PCC, for which we derive the tail dependence parameter. The hyperbolic distribution is specified by the following probability density:
\begin{equation}\label{eqPdfHyp}
f_{\rm HB}(x)= \frac{\sqrt{\psi/\chi}}{2 \alpha K_1(\sqrt{\chi\psi})}e^{-\alpha \sqrt{\chi+(x-\mu)^2}+\beta(x-\mu)}.
\end{equation}
Here, $\psi = \alpha^2-\beta^2$ with $\alpha>\beta$  and $K_{\lambda}$ is a modified Bessel function of the second kind, $\mu$ is a location parameter, $\chi$ is a scale parameter, $\alpha + \beta$ describes the lower tail, and $\alpha -\beta$ the upper tail. The expressions for the mean and variance of the hyperbolic distribution are known analytically. The parameter $\mu$ is chosen such that the mean is zero. The variance $\Lambda$ is a monotonically increasing function of $\chi$. We denote the corresponding hyperbolic distribution with mean zero and variance $\Lambda$ as ${\rm HB}(\alpha,\beta,\Lambda)$. 

The DGP for $\bm{Y}$ is given by the linear transformation $\bm{Y}=\bm{W}\bm{P}$ and the following distributions for the generators:
\begin{equation} \label{eqExPccHn}
P_1 \sim  {\rm HB}(\alpha_1,\beta_1,\Lambda_1), \quad P_j \sim N (0,\Lambda_j) \quad  {\rm for} \quad 1<j \leq d.
\end{equation}
We start by considering the two-dimensional case for simplicity. Then, the second PC is normally distributed as $P_2 \sim N(0,\Lambda_2)$ with $\Lambda_2=d-\Lambda_1$. Suppose there is a positive linear correlation $\rho$ between the risk factors $Y_i$, such that $0<\rho<1$, and we apply PCA to the correlation matrix of $\mathbf{Y}$. The first PC weight vector is then the parallel movement of the risk variables $\mathbf{w}_{1}= [1,1]^T/\sqrt{2}$, while the second PC weight vector is the anti-parallel movement $ \mathbf{w}_{2}= [1,-1]^T/\sqrt{2}$. A negative correlation would cause the anti-parallel vector to be the first PC. Moreover, we have $\Lambda_2=d-\Lambda_1$ with $d=2$. This means we can parametrize the implicit copula by either $\rho$ or $\Lambda_1$. A contour plot of the Hyperbolic-Normal copula density with Gaussian marginal distributions is given in Figure \ref{figPccDens}. We prove the following proposition for the tail dependence of the two-dimensional HB-N PCC.   

\begin{prop}\label{propTail}
 Consider two independent PCs with a hyperbolic and a normal distribution: $P_1\sim {\rm HB}(\alpha_1,\beta_1,\Lambda_1)$ and $P_2\sim {\rm N}(0,\Lambda_2)$. Consider the risk factors $Y_1 = (P_1+P_2)/\sqrt{2}$ and $Y_2 = (P_1-P_2)/\sqrt{2}$. Then, the copula of $Y_1$ and $Y_2$ has a lower tail dependence parameter given by $\eta_l=2\Phi(-(\alpha_1 + \beta_1)\sigma_2)$ with $\sigma_2=\sqrt{\Lambda_2}$.
\end{prop}

\begin{proof}
We start by noting that the lower tail of the hyberbolic density is governed by:
$f_{P_1}(x)\propto e^{(\alpha_1 +\beta_1)x}$. We proceed by evaluating the expression for the (lower) tail dependence coefficient $\eta_l$ \citep{McNeil2005,Joe2014}:
\begin{eqnarray}\label{eqTailCoef}
\eta_l &=& \lim_{q \downarrow 0} P(U_1 < q \lvert U_2 < q) 
= \lim_{y_q \rightarrow -\infty} \{ P(Y_1 < y_q \vert Y_2 = y_q) + P(Y_2 < y_q \vert Y_1 = y_q)\} \nonumber \\&=&\lim_{y_q \rightarrow -\infty} \frac{\int_{-\infty}^{y_q} f_{\bm{Y}}(y_1 , y_q)dy_1}{f_{Y_2}(y_q)}+\lim_{y_q \rightarrow -\infty} \frac{\int_{-\infty}^{y_q} f_{\bm{Y}}(y_q , y_2)dy_2}{f_{Y_1}(y_q)}
\end{eqnarray}
Next, we evaluate the conditional probabilities, which are equal when the variables are exchangeable:
\begin{eqnarray}\label{eqTailDep}
\eta_l&=& \lim_{y_q \rightarrow -\infty} \frac{2\int_{-\infty}^{y_q} f_{\bm{Y}}(y_1,y_q)dy_1}{f_{Y_2}(y_q)}=\lim_{y_q \rightarrow -\infty} \frac{2\int_{-\infty}^{y_q} f_{P_1}\left(\frac{y_q+y_1}{\sqrt{2}}\right)f_{P_2}\left(\frac{y_1-y_q}{\sqrt{2}}\right)dy_1}{\int_{-\infty}^{\infty} \sqrt{2}f_{P_1}(p_1)f_{P_2}\left(p_1-\sqrt{2}y_q \right)dp_1}\nonumber \\[0.5ex]
&=& \lim_{y_q \rightarrow -\infty} \frac{2\int_{-\infty}^{y_q} e^{(\alpha_1 + \beta_1) (y_1+y_q)/\sqrt{2}} e^{-(y_1 - y_q)^2/(4\sigma_2^2)} dy_1}{\sqrt{2}\int_{-\infty}^{\infty}e^{(\alpha_1 + \beta_1) p_1} e^{-(p_1 -\sqrt{2} y_q)^2/(2\sigma_2^2)} dp_1}= \lim_{y_q \rightarrow -\infty} \frac{2 I_1}{\sqrt{2}I_2}  
\end{eqnarray}
In the third step, the leading contribution of the integrals arises from the regions centered around $y_q$ and $\sqrt{2} y_q$, respectively. The two integral expressions ($I_1$ and $I_2$) are standard integrals, for which we obtain:
\begin{eqnarray}\label{eqIntegrals}
I_1 &=&  2\sqrt{\pi}\sigma_2 e^{\sqrt{2}(\alpha + \beta)y_q+(\alpha + \beta)^2\sigma_2^2/2}\Phi(-(\alpha + \beta)\sigma_2) \nonumber\\
I_2 &=&\sqrt{2} \sqrt{\pi}\sigma_2 e^{\sqrt{2}(\alpha + \beta)y_q+(\alpha + \beta)^2\sigma_2^2/2}
\end{eqnarray}
We finalize by inserting the expressions from Eq. (\ref{eqIntegrals}) into Eq. (\ref{eqTailDep}).
\end{proof}

The tail parameter $\eta_l$ is seen to approach 1, when $(\alpha_1 + \beta_1)$ goes to zero, leading to a fatter lower tail for $P_1$. Similarly, we can show that the upper tail dependence parameter is given by $\eta_u=2\Phi(-(\alpha_1 - \beta_1)\sigma_2)$. The copula allows for different upper and lower tail dependence. The first common factor that drives the dependence of the joint downward shocks is skewed and leptokurtic, while the higher component is based on the normal distribution. Therefore, the additional complexity in the copula is added where it is most relevant. We show in applications that this copula is readily generalized to high dimensions and that it performs well for aggregate measures based on historic market data.

\begin{rem}  
For an elliptical distribution, tail dependence is generated when the random vector is regularly varying \citep{Hult2002}. For a PCC, we have shown explicitly that we can also generate tail dependence with exponentially and Gaussian tailed variables, since there is no elliptical symmetry assumed.
\end{rem}

\subsubsection{Skew $t_1$-$t_{d-1}$ PCC}

More general PCC constructions are obtained with the generalized hyperbolic distribution (GH) as generator. The GH distribution has a stochastic representation as a normal mean-variance mixture with the generalized inverse Gaussian (GIG) distribution as the stochastic mixing variable \citep{McNeil2005}. The corresponding class of implicit copulas is called a normal mixture copula. We will consider PCCs with multiple stochastic mixing factors, which will result in convolutions of the GH distribution. We start with the density for the $d$-dimensional GH distribution, for which various parameterizations exist \citep{Prause1999}. We use:
\begin{equation}
f_{\rm GH}(\bm{x}) =\frac{(\psi/\chi)^{\lambda/2}\alpha^{d/2-\lambda}e^{(\bm{x}-\bm{\mu})\bm{\beta}}}{(2\pi)^{d/2}\vert\bm{\Sigma}\vert^{1/2}K_{\lambda}(\sqrt{\chi\psi})}\frac{K_{\lambda-d/2} \left(\alpha \sqrt{\chi+(\bm{x}-\bm{\mu})'\bm{\Sigma}^{-1}(\bm{x}-\bm{\mu})}\right)}{\left(\sqrt{\chi+(\bm{x}-\bm{\mu})'\bm{\Sigma}^{-1}(\bm{x}-\bm{\mu})}\right)^{d/2-\lambda}} \nonumber
\end{equation}
Here, we have introduced the notations $\bm{\beta}=\bm{\Sigma}^{-1}\bm{\gamma}$ and $\alpha = \sqrt{\psi+\bm{\gamma}'\bm{\Sigma}^{-1}\bm{\gamma}}$, while $\lambda \in \mathbb{R}$, $\chi > 0$, $\psi > 0$, $\bm{\mu} \in \mathbb{R}^d$, $\bm{\gamma} \in \mathbb{R}^d$ and $\bm{\Sigma} \in \mathbb{R}^{d\times d}$ a positive definite matrix. Importantly, also the characteristic function for the GH distribution is analytically known \citep{Prause1999}, which is given by:
\begin{equation}\label{eqCfGh}
\phi_{\rm GH}(\bm{t})= \frac{e^{i\bm{t}'\bm{\mu}}\psi^{\lambda/2}}{(\psi+\bm{t}'\bm{\Sigma} \bm{t}-2i\bm{t}'\bm{\gamma})^{\lambda/2}} \frac{K_{\lambda} \left(\sqrt{\chi(\psi+\bm{t}'\bm{\Sigma}\bm{t}-2i\bm{t}'\bm{\gamma})}\right) }{K_{\lambda}(\sqrt{\chi\psi})}.
\end{equation}
Many relevant distributions in finance are part of the GH family, since $\lambda=-1/2$ results in the normal-inverse Gaussian (NIG) distribution, while  $\lambda=(d+1)/2$ results in the hyperbolic distribution. When $\lambda=-\nu/2$, $\chi=\nu$ and $\psi\rightarrow 0$, the multivariate skew $t$ distribution is obtained. The corresponding implicit copula is known as the GH skew $t$ copula, which is frequently used in financial applications \citep{Demarta2005,Lucas2014,Opschoor2020,Oh2023}.  When the skewness parameter vector $\bm{\gamma}$ goes to zero, the multivariate $t$ distribution and the $t$ copula with degrees-of-freedom (DoF) parameter $\nu$ are retrieved \citep{Demarta2005}. 

\begin{figure}
\centering
\includegraphics[width=0.75\textwidth]{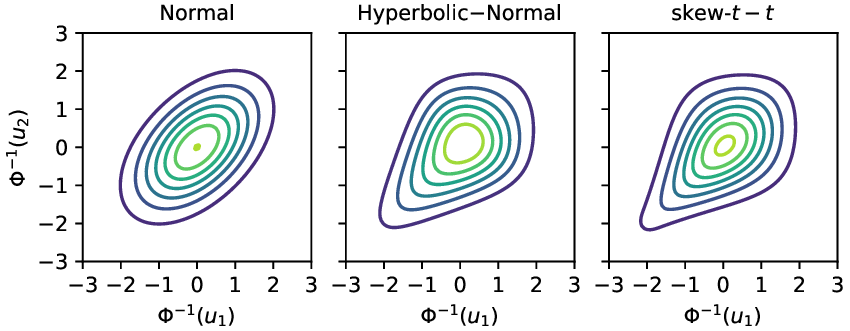}
\caption{Contour plots for bivariate copula densities generated by the Hyperbolic-Normal and the Skew $t_1$-$t_1$ Principal Component Copulas. For the visualization, we choose for the standard normal distribution as marginal distribution, so that the Gaussian copula results in elliptical contours. Because the PCC examples have a skewed and leptokurtic distribution along the first Principal Component, enhanced probabilities of joint downward movements are generated.}\label{figPccDens}
\end{figure}

Next, we explain a specific PCC construction based on GH distributions. An important feature of the copula construction is that the stochastic mixing variable $V$ for the first generator $P_1$ differs from the stochastic mixing variable $W$ for the higher generators $P_j$. We focus on the subfamily of the (skew) $t$ distribution, for which the stochastic mixing variable simplifies to the inverse gamma (IG) distribution \citep{Demarta2005}. The PCC generators therefore have the following stochastic representation:
\begin{equation}\label{eqNmvmix}
P_j=\begin{cases}
\mu_1 + \gamma_1 V +\sigma_{1} \sqrt{V} Z_1, & \text{ for } j=1\\
\sigma_{j}\sqrt{W} Z_j, & \text{ for } j > 1
\end{cases}
\end{equation}
Here,  $Z_j \sim N(0,1)$, $V\sim {\rm IG}(\nu_1/2,\nu_1/2)$ and $W\sim{\rm IG}(\nu_{>}/2,\nu_{>}/2)$ are independent stochastic variables, while $\sigma_{j} = \sqrt{\Sigma_{j,j}}$, $\nu_1$ is the DoF parameter for the first generator and $\nu_{>}$ the parameter for the higher generators. As a result, the generators $P_j$ are uncorrelated, but the higher PCs with $j>1$ are dependent due to the common mixing variable $W$. Since the expectation value and the variance of the GH distribution have analytic expressions \citep{Prause1999}, we can determine $\mu_1$ and $\Sigma_{j,j}$ such that for all $j$ the mean of $P_j$ is zero and the variance of $P_j$ equals $\Lambda_{j}$. For the GH skew $t$ distribution, these analytic expressions are particularly simple \citep{Demarta2005}, so that we obtain:
\begin{equation}
\mu_1 = -\frac{\nu_1\gamma_1}{(\nu_1-2)}, \quad \Sigma_{1,1} = \frac{(\nu_1-2)}{\nu_1}\Lambda_{1}-\frac{2\nu_1}{(\nu_1-2)(\nu_1-4)} \gamma_1^2. 
\end{equation}
Similarly, we have $\Sigma_{j,j} = (\nu_{>}-2)\Lambda_{j}/\nu_{>}$ for the higher generators with $j>1$. We require $\nu_1>4$ and $\nu_{>}>2$ for the variance $\Lambda_j$ to remain finite. This is because the heaviest tail of the skew $t$ distribution decays as $\vert x \vert^{-\nu/2-1}$ \citep{Aas2006}, while the $t$ distribution decays as  $\vert x \vert^{-\nu-1}$. 

The first copula generator, $P_1$, is used to generate asymmetric tail dependence. The higher PCs, $\bm{P}_{>}$, do not incorporate skewness. The Skew $t_{1}$-$t_{d-1}$ PCC is then based on the following specification: 
\begin{equation} \label{eqExPccSkewMvt}
P_1 \sim  {\rm Skew } \, t_1(\nu_1,\Lambda_{1},\gamma_1), \quad  \bm{P}_{>} \sim t_{d-1}(\nu_{>},\bm{\Lambda}_{>}). 
\end{equation}
Here, Skew $t_1(\nu_1,\Lambda_{1},\gamma_1)$ denotes the one-dimensional GH skew $t$ distribution with zero mean, variance $\Lambda_{1}$, degrees-of-freedom parameter $\nu_1$ and skewness parameter $\gamma_1$, while $t_{d-1}(\nu_{>},\bm{\Lambda}_{>})$ denotes the ($d-1$)-dimensional $t$ distribution with diagonal covariance matrix $\bm{\Lambda}_{>}$. The matrix $\bm{\Lambda}_{>}$ is obtained from $\bm{\Lambda}$ by omitting the first row and column containing $\Lambda_1$. The DGP for $\bm{Y}$ follows from the linear transformation $\bm{Y}=\bm{W}\bm{P}$. A contour plot of the Skew $t_1$-$t_{d-1}$ copula density with Gaussian marginal distributions is given in Figure \ref{figPccDens}. The pairwise tail dependence parameters, that follow from the convolution of two independent factors with regularly varying tails, have been analytically obtained by several authors \citep{Bouchaud2003, Malevergne2004, Oh2017}. Their results show that nontrivial tail dependence structures are obtained from the convolution of heavy-tailed factors with a common tail index.  

We conclude that the proposed implicit copula has the following combination of favorable properties specific to PCCs. It is skewed along the first generator, which is the main common factor. The distribution is elliptical in the subspace of the higher generators, which are the orthogonal directions to the main factor resulting in diversification. The stochastic structure is based on two stochastic mixing variables. This differs from the commonly applied (skew) $t$ copula, which has a single stochastic mixing factor $W$. In that case, the stochastic variance factor $W$ grows for all directions simultaneously in extreme scenarios. The variation orthogonal to the parallel movement is not a priori expected to grow during a crash, because it enhances diversification. Consequently, we find in our empirical study of Section \ref{secApplications} that the Skew $t_1$-$t_{d-1}$ PCC specification leads to improvements in terms of information criteria and in capturing systemic risk compared to other copulas. 

\subsection{Copula density}\label{subsecIndPcc}

Having discussed specific examples, we return to the general concept behind the copula construction technique, which is to generate novel dependence structures based on convolutions of analytically tractable probability distributions. In Fig. \ref{figPccTrans}, we illustrate the Data Generating Process of the Principal Component Copula. Characteristic functions are convenient in obtaining analytic results for convolutions of stochastic variables. We summarize the most relevant properties of characteristic functions that we require in the following lemma. For clarity, we distinguish three different cases, where the first two situations are special versions of the more general third case. In the first case, all PCs $P_j$ are independent, so that all (univariate) characteristic functions $\phi_{P_j}$ multiply. In the second case and third case, we have that subgroups of generators ${\bf P}_g$ are independent, so that the corresponding (multivariate) characteristic functions $\phi_{{\bm P}_g}$ multiply. The PCs within such a subgroup are then not necessarily independent.

\begin{figure}
\centering
\includegraphics[width=0.75\textwidth]{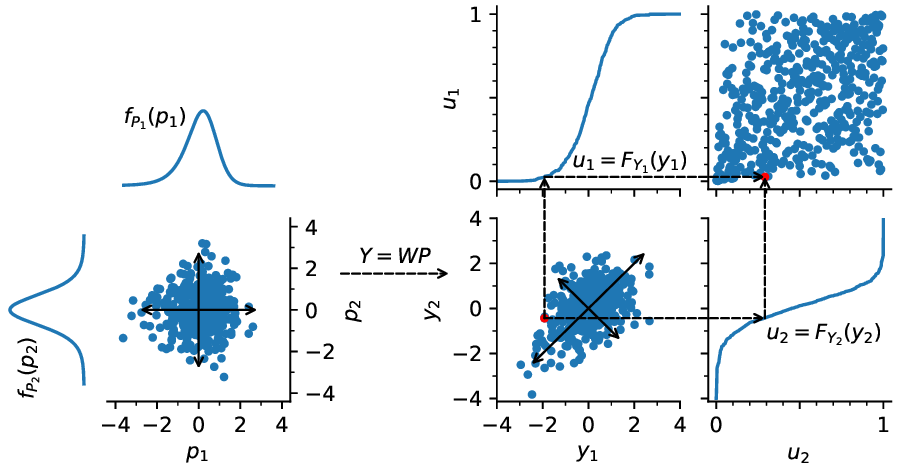}
\caption{Visualization of the data generating process for the Principal Component Copula (PCC). The joint distribution of uncorrelated generators $P_j$ forms the basis of the implicit copula. The orthogonal matrix $\bm{W}$ performs a rotation to introduce correlations between stochastic variables $Y_i$. Each marginal distribution $F_{Y_i}$ can be computed using convolution to complete the copula construction for $U_i$. Tail dependence is generated in a flexible and controlled way by distribution choices in the parallel and orthogonal directions.} \label{figPccTrans}
\end{figure}

\begin{lem} \label{lemChar}
Consider two $d$-dimensional stochastic vectors $\bm{Y}$ and $\bm{P}$ that are related to each other through the linear transformation $\bm{Y}=\bm{WP}$ with $\bm{W}$ a $d$-dimensional orthogonal matrix. Next, consider three cases for the copula generators ${\bm P}$ with increasing generality: 1) all generators $P_j$ are independent, 2) a single generator $P_1$ is independent, or 3) $m$ subgroups ${\bf P}_g$ of generators with $\bm{P}=(\bm{P}'_1,..,\bm{P}'_m)'$ are independent (case 3). \\ 
{\bf Case 1}. Let the characteristic function of each individual independent copula generator ${P}_j$ be known and given by $\phi_{P_j}(t_j)$. Then, we obtain for $\phi_{\bm Y}$ that: 
\begin{equation}
\phi_{\bm Y}(\bm t)= E[e^{i\bm{t}'\bm{Y}}]=E[e^{i\bm{t}'\bm{WP}}]=\prod_{j=1}^d E[e^{i (\bm{t}' \bm{w}_{j})P_j}]=\prod_{j=1}^d \phi_{P_j}(\bm{w}_{j}'\bm{t}).
\end{equation}
Here, $\bm{w}_{j}$ is the $j$-th column of $\bm{W}$.\\
{\bf Case 2}. Let the characteristic function of the first copula generator $P_1$ and the higher copula generators ${\bf P}_{>}$ be known and given by $\phi_{P_1}(t_1)$ and $\phi_{\bm{P}_{>}}(t_2,...,t_d)$. Then, we obtain for $\phi_{\bm Y}$ that: 
\begin{equation}
\phi_{\bm Y}(\bm t)= E[e^{i\bm{t}'\bm{WP}}]= E[e^{i(\bm{t}'\bm{w_1})P_1}]E[^{i\bm{t}'\bm{W}_{>}\bm{P}_{>}}]= \phi_{P_1}(\bm{w}'_{1} \bm{t})  \phi_{\bm{P}_{>}}(\bm{W}'_{>} \bm{t}).\nonumber
\end{equation}
Here, $\bm{w}_{1}$ is the first column of $\bm{W}$, while $\bm{W}_{>}$ is the $d \times (d-1)$ matrix by omitting $\bm{w}_{1}$ from $\bm{W}$.\\
{\bf Case 3}. Let the characteristic function of each $d_g$-dimensional generator ${\bf P}_{g}$ be known and given by $\phi_{\bm{P}_{g}}(\bm{t}_{g})$ with $\sum_{g=1}^m d_g =d$. Then, we obtain for $\phi_{\bm Y}$ that:  
\begin{equation} \label{eqCharY}
\phi_{\bm Y}(\bm t)= E[e^{i\bm{t}'\bm{WP}}]=\prod_{g=1}^m E[e^{i \bm{t}' \bm{W}_{g}\bm{P}_g}]=\prod_{g=1}^m\phi_{{\bm P}_g}(\bm{W}'_{g} \bm{t}) .
\end{equation}
Here, $\bm{W}_{g}$ is the $d \times d_g$ matrix by selecting the columns $j$ from $\bm{W}$ belonging to $P_j$ in subgroup $g$.\\
The marginal characteristic function $\phi_{Y_1}(t_1)$ is given by $\phi_{\bm Y}(t_1,0,...,0)$. In general, we obtain the marginal characteristic function $\phi_{Y_i}(t_i)$ for $Y_i$ from $\phi_{\bm Y}({\bm t})$ by setting all $t_j$ to zero for $j\neq i$, which we write as: 
\begin{equation} \label{eqCharYi}
\phi_{Y_i}(t_i)= \phi_{\bm Y}(t_i\bm{e}_i).
\end{equation}
Here, $\bm{e}_i$ denotes the $i$-th unit vector of dimension $d$.
\end{lem}

This lemma therefore makes use of two convenient properties of characteristic functions, namely that the characteristic functions of (groups of) independent variables multiply, and that the marginal characteristic function is obtained from the multivariate characteristic function by setting all other arguments to zero. 
The first case of independent generators is relevant for the hyperbolic-normal PCC, because for the multivariate normal distribution it is equivalent to be uncorrelated and independent. The second case is relevant for the Skew $t_1$-$t_{d-1}$ PCC, because then all higher PCs are dependent. The third case allows general splitting of the PCs in subgroups and encompasses the first two cases. As a result, the multivariate distribution properties of the implicit copula for ${\bm Y}$ become accessible from the properties of the generator distributions. We can use the analytic results for the characteristic function to derive the marginal density and the marginal distribution function for $Y_i$. These expressions are needed to determine the convolution copula density from Eq. (\ref{eqCopDens}) and for simulating from the copula. This is summarized in the following lemma that follows from the properties of the characteristic function and the Gil-Pelaez theorem:
\begin{lem} \label{lemDist}
The marginal probability density $f_{Y_i}$ follows from the marginal characteristic function $\phi_{Y_i}(t)$ using:
\begin{equation} \label{eqCharPdf}
f_{Y_i}(y_i)= \frac{1}{2\pi}\int_{-\infty}^{\infty} \phi_{Y_i}(t)e^{-ity_i} dt.
\end{equation}
The marginal distribution function $F_{Y_i}$ follows from the Gil-Pelaez theorem and is given by:
\begin{equation}\label{eqCharCdf}
F_{Y_i}(y_i)= \frac{1}{2}-\frac{1}{\pi}\int_0^{\infty} \frac{{\rm Im}[e^{-ity_i}\phi_{Y_i}(t)]}{t} dt. 
\end{equation}
\end{lem}

Next, we use Lemma \ref{lemChar} and \ref{lemDist} to prove the following proposition for the copula density of a PCC. We consider the most general third case of Lemma \ref{lemChar}, from which the other two special cases follow.

\begin{prop}
Consider a $d$-dimensional Principal Component Copula $C_{\bm{Y}}$ for the stochastic vector $\bm{Y}$ that is generated by the (uncorrelated) stochastic vector $\bm{P}$ using  $\bm{Y}=\bm{WP}$ with $\bm{W}$ an orthogonal transformation matrix. Let $\bm{P}=(\bm{P}'_1,...,\bm{P}'_m)'$ be split in independent $d_g$-dimensional stochastic vectors $\bm{P}_g$ with known probability densities $f_{\bm{P}_g}$ and characteristic functions $\phi_{\bm{P}_g}$. Then, the $d$-dimensional copula density can be expressed in terms of $f_{\bm{P}_g}$ and $\phi_{\bm{P}_g}$ using only one-dimensional integrals and one-dimensional inversions.
\end{prop}

\begin{proof}
We start from Eq. (\ref{eqCopDens}) for the continuous probability density of an implicit copula. We consider the denominator in Eq. (\ref{eqCopDens}), for which we require the expressions $F_{Y_i}^{-1}$ and $f_{Y_i}$. Combining the results from Eqs. (\ref{eqCharY}) and (\ref{eqCharYi}) with Eqs. (\ref{eqCharPdf}) and (\ref{eqCharCdf}), we obtain $F_{Y_i}$ and $f_{Y_i}$ as one-dimensional integrals over characteristic functions $\phi_{\bm{P}_g}$. Consequently, $F_{Y_i}^{-1}$ can be obtained by one-dimensional inversion. Next, we proceed with the expression for the numerator. We have that $\bm{P}=\bm{W}'\bm{Y}$ and that $\bm{P}_g=\bm{W}_g'\bm{Y}$ with $\bm{W}_{g}$ the $d \times d_g$ matrix containing the columns $j$ from $\bm{W}$ belonging to $P_j$ in subgroup $g$. As a result, we have:  
\begin{equation}\label{eqDensY}
f_{\bm{Y}}(\bm{y})= f_{\bm{P}}(\bm{W}'\bm{y})\lvert J \rvert 
=\prod_{g=1}^m f_{\bm{P}_g}(\bm{W}_g'\bm{y}),
\end{equation}
where $f_{\bm{P}_g}$ denotes the probability density of $\bm{P}_g$. In the second step, we used the independence of the PC groups and that the Jacobian determinant $J$ equals $\pm 1$ for an orthogonal transformation. Since we have expressed $f_{Y_i}$, $F_{Y_i}$ and $f_{\bm{Y}}$ in terms of $f_{P_g}$ and $\phi_{P_g}$ using only one-dimensional integrals, and since $F^{-1}_{Y_i}$ is obtained by one-dimensional inversion, we have achieved the desired result for the copula density from Eq. (\ref{eqCopDens}).
\end{proof}

The integrals in Eqs. (\ref{eqCharPdf}) and (\ref{eqCharCdf}) can be efficiently evaluated numerically using the Fast Fourier Transform \citep{Carr1999}. Alternatively, the density and distribution function can be accurately approximated with the COS-method, which is based on cosine expansions \citep{Fang2009}. Since we only need to evaluate one-dimensional integrals to determine the density from Eq. (\ref{eqCopDens}) for an PCC, it is feasible to compute the full copula density in high dimensions. This is useful for example to perform maximum likelihood estimation (MLE). In Section \ref{secApplications}, we give numerical examples. 

\begin{rem} 
For a factor copula based on multiple factors, we have that Eq. (\ref{eqDensY}) does not apply, because it is based on a latent factor structure \citep{Oh2017}, rather than an orthogonal transformation. As a result, determining $f_{\mathbf{Y}}$ and $c(u_1,...,u_d)$ requires integrating over all common factors in that case. 
\end{rem}

We can specify the PCC ($C_{\bm{\Lambda}, \mathbf{W}, \bm{\alpha}}$) by the eigenvalues $\Lambda_j$ which describe the variance of the PCs, by the orthogonal matrix elements $w_{i,j}$ describing the linear transformation from $P_j$ to $Y_i$, and by the shape parameters $\bm{\alpha}$, which allow $P_j$ to be skewed or heavy-tailed. By definition of the PCC, we require the parameters $\Lambda_j$ and $w_{i,j}$ to form a valid correlation matrix, since they satisfy $\bm{\rho}_{\bm{Y}}=\bm{W} \bm{\Lambda} \bm{W}'$. This leads to restrictions on admissible parameter values for $\Lambda_j$ and $w_{i,j}$, which can be cumbersome to keep track of while performing estimation in high dimensions. As a result, we prefer to parametrize the PCC by the correlation matrix $\bm{\rho}_\mathbf{Y}$ and the shape parameters $\bm{\alpha}$. In high-dimensional applications, we prefer to estimate $\bm{\rho}_\mathbf{Y}$ with a moment estimator, such that it is automatically a valid correlation matrix, after which we perform the decomposition in eigenvalues and eigenvectors. If the eigenvalues $\Lambda_j$ are non-degenerate, the eigenvalue decomposition of $\bm{\rho}_\mathbf{Y}$ is unique up to the sign for the eigenvectors. In practice, it is common to enforce a sign convention, for example, that the largest eigenvector weight is taken to be positive. 

We note that the PCC is not equivalent to merely specifying a copula between PCs, since this is not sufficient to fully specify the copula for the stochastic variables $Y_i$. For example, an independence copula between PCs can lead to strong tail dependence between risk factors, when the PCs have fatter tails than Gaussian (see Proposition \ref{propTail}). The PCC fully specifies the implied copula density for the stochastic variables $Y_i$ that is generated by the PCs $P_j$. To the best of our knowledge, the implicit copula that is generated by PCA has not been systematically studied before.

\section{Estimation procedures} \label{secEst}

Having introduced the PCC and considered several properties, we next discuss the estimation procedures. In case of financial returns, it is common to first apply a time series filter (e.g. GARCH) to arrive at devolatized returns, after which the dependence structure is studied. In Section \ref{subsubsecTimeSeries}, we give a concrete example for such a time series model.  In this section, we assume that a $d$-dimensional data sample is observed that consists of $n$ iid observations $\bm{X}_{t}$ with $t=1,...,n$.

\subsection{Data exploration and model initiation} 

\begin{figure}
\centering
\includegraphics[width=0.95\textwidth]{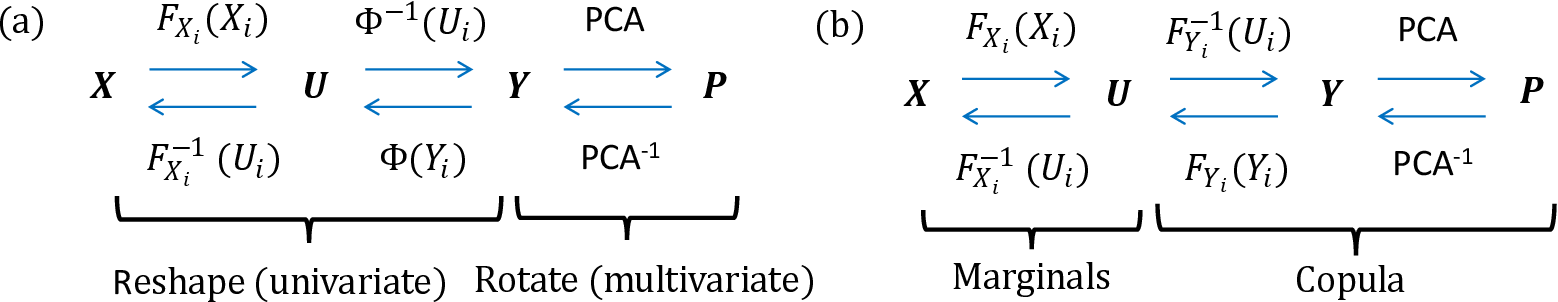}
\caption{Modelling scheme for: (a) Principal Component Analysis based on normal scores, and for: (b) a general Principal Component Copula.}\label{figSchemes}
\end{figure}

We start with explaining a procedure to explore the data sample, which is also convenient to initialize estimation procedures. We follow the scheme of Fig. \ref{figSchemes}a, which amounts to performing PCA on normal scores. The advantage of this scheme is that the first three steps can be performed directly on the original data $\bm{X}_{t}$ without requiring knowledge about the parameters of the implicit copula. Since we are only interested in the copula for $\bm{X}_{t}$, we transform the original data $X_{i,t}$ in the first step of Fig. \ref{figSchemes}a to (pseudo) copula observations using $U_{i,t}= {\rm rank}(X_{i,t})/(n+1)$. The second step of Figure \ref{figSchemes}a is also readily performed by obtaining the normal scores $Y_{i,t}=\Phi^{-1}(U_{i,t})$. As a result, the data is symmetrized and outliers are reduced before applying PCA. The third step is to apply PCA by determining the eigenvalues and eigenvectors of the sample correlation matrix of $\bm{Y}$ and decompose this matrix in the diagonal eigenvalue matrix $\bm{\Lambda}$ and the orthogonal matrix $\bm{W}$. Using $\bm{P}_t = \bm{W}'\bm{Y}_t $ we obtain an initial approximation of the PC observations $P_{j,t}$ based on matrix multiplication of normal scores, which is easily performed.  

We have now rescaled and rotated the data, while automatically detecting the most important directions based on normal scores. If the implicit copula is Gaussian, it is valid to use $F_{Y_i}=\Phi$ as marginal distribution function. However, this does not hold in general. For the $t$ copula, the appropriate marginal distribution would be the univariate $t$ distribution function. In general, we require the use of $F_{Y_i}$ based on Eqs.~(\ref{eqCharYi}) and (\ref{eqCharCdf}) in the second step of Fig.~\ref{figSchemes}b to obtain a consistent implicit copula representation generated by stochastic PCs. For data exploration and model initiation, we still find it relevant to visualize the properties of $\Phi^{-1}(U_{i,t})$ with Fig.~\ref{figSchemes}a, also in the non-Gaussian case. For example, if the true copula is a $t$ copula, non-elliptical deformations are observed when visualizing $\Phi^{-1}(U_{i,t})$. Since the use of normal scores allows to detect deviations from the Gaussian copula, like asymmetry and tail dependence, it provides guidance in exploring appropriate copula specifications. Moreover, we find that the correlation matrix of normal scores also leads to accurate initializations for $\Lambda$ and $\bm{W}$ when starting the consistent copula estimation schemes we discuss next.  

\subsection{Maximum likelihood estimation} \label{subsecMl}

The self-consistent scheme of Fig. \ref{figSchemes}b illustrates the implicit copula based on PC generators for which we derived the copula density in Section \ref{subsecIndPcc}. Next, we explain maximum likelihood estimation (MLE) of the PCC. We start from the general log-likelihood based on the copula density from Eq. (\ref{eqCopDens}), which results in:
\begin{eqnarray}\label{eqLogL}
\ell(\boldsymbol{\theta}; U_{1,t},..,U_{d,t})&=&\sum_{t=1}^n {\rm ln}\, c(U_{1,t},...,U_{d,t};\boldsymbol{\theta})\nonumber\\&=& \sum_{t=1}^n {\rm ln}\, f_{\mathbf{Y}}(Y_{1,t},...,Y_{d,t};\boldsymbol{\theta}) -\sum_{t=1,i=1}^{n,d} {\rm ln}\, f_{Y_i}(Y_i;\boldsymbol{\theta}).
\end{eqnarray}
Here, $Y_{i,t}=F^{-1}_{Y_i}(U_{i,t};\boldsymbol{\theta})$, while $f_{Y_i}$, $F_{Y_i}$ and $f_{\mathbf{Y}}$ are given by Eqs. (\ref{eqCharPdf}), (\ref{eqCharCdf}) and (\ref{eqDensY}). The first two equations require the marginal characteristic function that is obtained from Eqs. (\ref{eqCharY}) and (\ref{eqCharYi}). The MLE is given by:
\begin{equation}\label{eqMlEst}
\hat{\boldsymbol{\theta}}=\underset{\boldsymbol{\theta} \in \boldsymbol{\Theta}}{\rm arg \,\, max \,\,} \ell(\boldsymbol{\theta}; U_{1,t},..,U_{d,t}).
\end{equation}
MLEs based on pseudo copula observations have the desirable asymptotic properties of being consistent, asymptotically normal and efficient, as shown by \cite{Genest1995}. The main challenge for the PCC is then to find a fast and accurate way to determine the copula density and the loglikelihood in Eq. (\ref{eqLogL}), which can consequently be used in optimization algorithms for Eq. (\ref{eqMlEst}). 

\subsubsection{Numerical approach} \label{subsecMlNum}

In order to perform the optimization from Eq. (\ref{eqMlEst}), we note that it is possible to efficiently evaluate Eq. (\ref{eqCharPdf}) and (\ref{eqCharCdf}) numerically with the Fast Fourier Transform, which has been used for option pricing in the finance literature by \cite{Carr1999}. Another efficient method from the option pricing literature is the COS method by \cite{Fang2009}, which is based on the following cosine expansion:
\begin{equation}\label{eqPdfCos}
f_{Y_i}(y) \simeq \frac{c_0}{2} +\sum^{N_c}_{k=1}c_k \cos\left( k\pi\frac{y-a}{b-a} \right).
\end{equation}
The expansion coefficients $c_j$ can be analytically expressed in terms of the characteristic function $\phi_{Y_i}$ as:
\begin{equation}
c_k = \frac{2}{b-a} {\rm Re} \left[\phi_{Y_i}\left(\frac{k\pi}{b-a}\right)e^{-i\frac{ka\pi}{b-a}}\right].
\end{equation}
For $F_{Y_i}$, a similar expansion in terms of sine functions can be performed: 
\begin{equation}\label{eqCdfCos}
F_{Y_i}(y) \simeq \frac{c_0}{2}(y-a)+\sum^{N_c}_{k=1}c_k \frac{b-a}{k\pi}\sin\left( k\pi\frac{y-a}{b-a} \right).
\end{equation}
Because $Y_i$ has unit variance and because the COS method converges exponentially, we typically find accurate results using: $a=-10,b=10$ and $N_c=100$. Having implemented the COS method, we can efficiently evaluate $f_{Y_i}(y_i)$ and $F_{Y_i}(y_i)$ on a coordinate grid of our choice. If we evaluate $F_{Y_{i}}(y_{i})=u_{i}$ on a sufficiently fine grid, we immediately also obtain $F^{-1}_{Y_{i}}(u_{i})=y_{i}$ on the same grid. Any other value for $F^{-1}_{Y_{i}}(u_{i})$ can be obtained by interpolation. An alternative method to sample from the inverse distribution function is stochastic collocation \citep{Grzelak2019}.

Next, the optimization of Eq. (\ref{eqMlEst}) is performed to obtain the ML estimates. The optimization is initialized with the linear correlation matrix of the normal scores $\Phi^{-1}(U_{i,t})$, which is the ML estimator for the correlation parameters of the Gaussian copula, and with plausible starting values of the shape parameters $\bm{\alpha}$. Gradient descent is then employed to update the parameters $\bm{\rho}_\mathbf{Y}$ and $\bm{\alpha}$. After each update of $\bm{\rho}_\mathbf{Y}$, we reperform eigenvalue decomposition (PCA) to obtain $\bm{\Lambda}$ and $\mathbf{W}$. The correlation matrices are kept positive semi-definite by requiring the eigenvalues of $\boldsymbol{\Lambda}$ to be nonnegative. 

Full ML optimization remains feasible in high-dimensions when additional structure is applied to the correlation matrix, so that it can be parameterized in a parsimonious way. In higher dimensions, with no additional predefined model structure, the number of correlation parameters quickly becomes very large for numerical optimization. Although a parsimonious model structure in itself is beneficial, it is usually not desirable to make a priori assumptions on the correlation matrix structure. Therefore, we propose an alternative estimation scheme in the next section.

\subsection{Generalized Method of Moments} \label{subsecGmm}

Besides ML, another popular class of estimators is the generalized method of moments (GMM) by \cite{Hansen1982}. It starts from a vector-valued function $\bm{g}(\bm{Y},\boldsymbol{\theta})$ with an expectation value of zero at the true parameter values $\boldsymbol{\theta}_0$. The expectation value of $\bm{g}$ can be estimated with the vector valued sample mean: 
\begin{equation}
\hat{\bm{m}}(\bm{\theta}) =  \frac{1}{n}\sum_{t=1}^n \bm{g}(\bm{Y}_t,\bm{\theta}).
\end{equation}
The GMM estimator is then given by:
\begin{equation}
\hat{\bm{\theta}} = \underset{\bm{\theta} \in \bm{\Theta}}{\rm arg \,\, min \,\,} \hat{m}(\bm{\theta})^T \bm{W}_g \hat{m}(\bm{\theta}).
\end{equation}
Here, $\bm{W}_g$ is a positive semi-definite weight matrix. The GMM estimator has the desirable asymptotic properties of being consistent and normally distributed. A suitable choice of the matrix $\bm{W}_g$ leads to the most efficient estimator based on the moment conditions \citep{Hansen1982}. 

In high dimensions, it is preferable to estimate the many correlation parameters with a moment estimator and the fewer shape parameters with an ML estimator. This is a common approach in financial applications when using a $t$-copula, where the correlation parameters are based on Kendall's tau estimates and the shape parameter $\nu$ using ML \citep{McNeil2005}. We develop a similar strategy for the PCC. First, we introduce the generalized moment equations for the scale parameters $\boldsymbol{\rho}_{\bm{Y}}$ and the shape parameters $\boldsymbol{\alpha}$, giving:
\begin{eqnarray}
    E[Y_i Y_j - \rho_{Y,i,j}]=E[F_{Y_i}^{-1}(U_i)F_{Y_j}^{-1}(U_j) - \rho_{Y,i,j}] &=& 0 \nonumber \\
   E[\nabla_{\bm{\alpha}} \ell(\bm{\rho}_{\bm{Y}}, \bm{\alpha}; U_{1},..,U_{d})] &=& 0 \label{EqCondMm}
\end{eqnarray}
In the second line, we have written the ML estimator as a GMM condition. This shows that an estimator based on the generalized moment conditions in Eq. (\ref{EqCondMm}) can be viewed as a GMM estimator.   

\subsubsection{Numerical approach} \label{subsecGmmNum}

Next, we explain how we numerically arrive at the solution of the GMM equations following from the conditions in Eq. (\ref{EqCondMm}). Note that $F_{Y_i}^{-1}$ depends itself on $\boldsymbol{\rho}_{\mathbf{Y}}$ and $\boldsymbol{\alpha}$, making the corresponding equations nonlinear. We propose an iterative procedure to solve the equations. The procedure initializes the correlation matrix $\hat{\rho}^{(0)}_{\bm{Y}}$ with the sample correlation of the normal scores $\Phi^{-1}(u_{i,t})$. Next, we decompose the correlation matrix into eigenvectors using PCA. In general, the following holds for the PCC: 
\begin{equation}
\rho_{Y,i,j}= E[Y_i Y_j]=\sum^d_{k,l=1} E[ w_{i,k}P_{k} w_{j,l} P_l]=\sum_{k=1}^d\Lambda_k w_{i,k}w^T_{k,j}. 
\end{equation}
Here, we used that $E[Y_i]=E[P_j]=0$, that $E[P^2_j]=\Lambda_j$ and that the PCs are uncorrelated. Having initialized the PC vectors $\hat{\bm{w}}_j^{(0)}$ and variances $\hat{\Lambda}_j^{(0)}$, we can determine the shape parameters $\boldsymbol{\alpha}$ using ML. To this end, we use Eqs. (\ref{eqCopDens}), (\ref{eqCharY}), (\ref{eqCharPdf}) and (\ref{eqCharCdf}) for the likelihood in combination with the COS method. An iterative algorithm can then solve the GMM equations, which we show in Algorithm \ref{algoPcc}. In order for the iterative algorithm to converge towards a fixed point, the usual conditions for contraction are required, i.e., that derivatives are smaller than one around the fixed point.

\begin{algorithm}[t]
\caption{Iterative algorithm to estimate the PCC with GMM \label{algoPcc}}
\begin{enumerate}[leftmargin=0.5cm]
    \item 
    Initialize scale parameters $\hat{\bm{\rho}}^{(0)}_{\bm{Y}}$ and shape parameters $\hat{\bm{\alpha}}^{(0)}$
    \item  
    For each recursion step $k=1,..,n_r$ update parameters as follows:
    \begin{enumerate}[leftmargin=0.75cm]
        \item Update correlation parameters given shape parameters:
            \begin{equation}
                    \hat{\rho}^{(k)}_{\mathbf{Y},i,j} = \frac{1}{n} \sum_{t} \hat{F}^{-1}_{Y_i} \left(U_{i,t};\hat{\rho}^{(k-1)}_{\bm{Y}}, \hat{\bm{\alpha}}^{(k-1)}\right) \hat{F}^{-1}_{Y_i} \left(U_{j,t};\hat{\rho}^{(k-1)}_{\bm{Y}}, \hat{\bm{\alpha}}^{(k-1)}\right) \nonumber
            \end{equation}       
        \item Perform PCA to update vector weights $\hat{w}_{i,j}^{(k)}$ and eigenvalues $\hat{\Lambda}_j^{(k)}$
        \item Update shape parameters given correlation parameters using ML:
            \begin{equation} 
            \hat{\bm{\alpha}}^{(k)}=\underset{\bm{\alpha}}{\rm arg \, max \,} \ell(\bm{\alpha}, \hat{\bm{W}}^{(k)},\hat{\bm{\Lambda}}^{(k)}; U_{1,t},..,U_{d,t}).\nonumber
            \end{equation}
    \end{enumerate}   
    \item Return estimated GMM parameters $\{\hat{\bm {\rho}}^{(n_r)}_{\bm{Y}}, \hat{\bm{\alpha}}^{(n_r)}\}$ 
\end{enumerate}
\end{algorithm}

 An advantage of the iterative approach is that it allows the correlation matrix to be determined with a moment estimator. This avoids having to numerically optimize over many correlation parameters, which is required for ML or Simulated Method of Moments when there is no closed-form expression for the correlation parameters in terms of the data available \citep{Oh2017}. As a result, these methods have been mainly applied to high-dimensional copulas with a reduced number of parameters by imposing additional structure to the correlation matrix. Although reducing complexity can be favorable, it is also beneficial to have a technique that avoids the necessity to make a priori assumptions on the correlation structure. In that case, a reduction in parameters can still be imposed, but only when it is deemed necessary in terms of estimation stability and accuracy.

In the next section, we show that the algorithm performs well in detecting PCs with large eigenvalues. In the studied examples, the estimation of the first eigenvalues and eigenvectors are found to be robust to the changes in shape parameters of the PC distributions. Consequently, the initialization steps already give accurate results, and only a few iteration steps are required to obtain the full solution. In the next section, we show good performance with the hybrid estimator of Algorithm \ref{algoPcc} for large copula models in high dimensions. 

\section{Applications}\label{secApplications}

We study the PCC and its estimators using simulated data and historical global market data.  

\subsection{Simulation study}\label{subsecSim}

We perform a simulation study for a 100-dimensional PCC as schematically visualized in Fig. \ref{figSchemes}b. When simulating, the scheme is followed from right to left. We start with generating simulations for $\bm{P}$ based on the specified generator family. After this, we transform to $\bm{Y}$ using the linear transformation $\bm{Y}=\bm{W}\bm{P}$. Finally, we apply $F_{Y_i}$ to obtain samples for the copula $\bm{U}$. The stochastic variables $Y_i$ have a mean equal to zero and a variance equal to $1$. Therefore, the first and second moments of $\bm{Y}$ are determined by the correlation matrix $\rho_{\bm Y}$. In our stylized example, the correlation matrix has the following structure:
\begin{equation} \label{eqCorrHighD}
\rho_{Y_i,Y_j} = \xi_{i}\xi_{j}  + \gamma_{i}\gamma_{j} 
\end{equation}
Here, we use the following expressions with $i$ ranging from 1 to $d$: 
\begin{equation}
\xi_{i} = \frac{2}{5}\left(1+e^{-i/d}\right),\quad \gamma_{i}=\frac{3}{5}\tanh(4i/d-2).
\end{equation}
This parametrization leads to $d(d-1)/2$ different correlation parameters with values ranging from $0.10$ to $0.96$. By performing PCA, we obtain the true eigenvector matrix $\bm{W}$ and the true eigenvalues $\bm{\Lambda}$ of $\rho_{\bm Y}$. The first two PC vectors are shown in Fig. \ref{figEigExample}. These vectors are seen to describe a parallel shock and an anti-parallel shock. The first PCs have resulting eigenvalues $\lambda_1=43.6$ and $\lambda_2=18.7$, respectively, which therefore explain 62\% of the variance. All higher PCs have eigenvalues less than 1, so that they do not carry relevant information.  

\begin{figure}
\centering
\includegraphics[width=0.7\columnwidth]{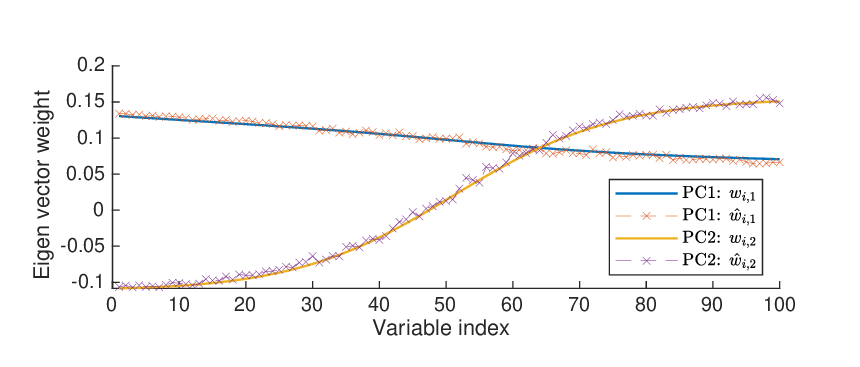}
\caption{True eigenvector weights $w_{i,j}$ and estimated eigenvectors weights $\hat{w}_{i,j}$ for the first two Principal Components $j \in \{1,2\}$ as a function of the variable index $i\in\{1,...,100\}$. }\label{figEigExample}
\end{figure}

The DGP for $\bm{Y}$ in our example is given by the linear tranformation $\bm{Y}=\bm{W}\bm{P}$ and the following distributions:
\begin{equation} \label{eqExPccHnSim}
P_j \sim  {\rm HB}(\alpha_j,\beta_j,\Lambda_j) \quad \forall j\in \{1,2\} , \quad P_j \sim N (0,\Lambda_j) \quad \forall j>2
\end{equation}
Here, ${\rm HB}(\alpha_j,\beta_j,\Lambda_j)$ denotes the hyperbolic distribution from Eq. (\ref{eqPdfHyp}), where $\mu$ and $\chi$ are taken such that the mean is equal to zero and the variance is equal to $\Lambda_j$. In Table \ref{tabEstHighD}, the parameters selected for the study are presented. The chosen parameters for the first two PC distributions give rise to skewness and tail dependence, due to the heavier tails of the hyperbolic distribution compared to the normal distribution. The parameters lead to lower tail dependence estimates between risk pairs in the range from 0.05 to 0.85. We consider a higher dimension and a larger number of shape parameters in the simulations than in the case study of Section \ref{subsecCase} in order to test the estimators under more challenging circumstances.

In the simulation study, we generate 1500 copula observations for each estimation, which is a similar number of observations to our empirical case study in Section \ref{subsecCase}. Following the specified DGP, we sample the PCs $\bm{P}$ independently from the hyperbolic and normal distributions with their true parameters $\alpha_j$, $\beta_j$ and $\Lambda_j$. Then, we transform the simulated data using $\bm{Y}=\bm{W}\bm{P}$ with $\bm{W}$ the true eigenvector matrix of $\rho_{\bm Y}$. Finally, we return $U_i = F_{Y_i}(Y_i)$ for the simulated copula observations. There are two different numerical approaches to obtain $F_{Y_i}$. The first method is to use Eqs. (\ref{eqCharYi}) and (\ref{eqCdfCos}), which converges exponentially \citep{Fang2009}. The second method is to simulate many observations for $Y_i$ and obtain a simulation-based empirical distribution function $\hat{F}_{Y_i}$ for $Y_i$. Due to its fast convergence, we use the first method. Having generated the copula data, we start with applying MLE for the 4 shape parameters $\hat{\alpha}_j$ and $\hat{\beta}_j$, taking the true eigenvector weights $w_{i,j}$ and eigenvalues $\Lambda_j$ as given. This hypothetical situation is used as performance reference. Next, the hybrid estimation procedure from Algorithm \ref{algoPcc} is applied, which estimates all correlation parameters and the 4 shape parameters jointly. We do not inform our estimation algorithm about the underlying parameter structure. 

In order to study the performance of the estimation procedure, 100 Monte Carlo replications of the full estimation procedure are generated and the mean and the standard deviation of the estimators are determined. In Table \ref{tabEstHighD}, the results for the four shape parameters using MLE and the results with Algorithm \ref{algoPcc} are presented. We also show the results for the first two eigenvalues $\hat{\Lambda}_j$. The estimation result for the corresponding eigenvectors $\hat{w}_{i,j}$ are shown in Fig. \ref{figEigExample}, which reveals that the eigenvector estimates are very accurate. From Table \ref{tabEstHighD}, it is seen that the MLEs for the shape parameters perform well in terms of bias and standard errors. This is expected, since the true eigenvector weights and eigenvalues were used and since MLE is consistent with minimal asymptotic variance. 

\begin{table}
\begin{center}
\caption{\label{tabEstHighD} Estimation results based on a simulation study for the PCC in 100 dimensions using 1500 simulated data points. The second column shows the true parameter values, the third column the MLEs for only the shape parameters (using the true values for the other parameters) and the fourth column the corresponding standard deviations based on 100 Monte Carlo replications of the estimation procedure. The fifth and sixth columns show the estimate and the standard deviations based on GMM using Algorithm \ref{algoPcc}.}
\vspace{0.3cm}
\begin{tabular}{ l l l l l l l}
\hline
Parameter       &  True &  ML    &        & &   GMM  &  \\
\hline 
 $\lambda_1$   &  43.61 &       &         & & 43.57 & (1.07)\\
 $\lambda_2$   &  18.70 &       &         & & 18.77 & (0.55)\\
 $\alpha_1$    &  0.50  & 0.47  &  (0.05) & & 0.47  & (0.06)\\
 $\beta_1$     &  -0.25 & -0.23 &  (0.04) & & -0.23 & (0.05)\\
 $\alpha_2$    &  1.00  & 0.99  &  (0.04) & & 1.02  & (0.07)\\
 $\beta_2$     &  0.25  & 0.26  &  (0.06) & & 0.26  & (0.07)\\
\hline 
\end{tabular}
\end{center}
\end{table}

We compare these ML shape estimates with the hybrid GMM estimator of Algorithm \ref{algoPcc}, where we need less than five iterations for satisfactory convergence in this example. The estimation results of Algorithm \ref{algoPcc} in the last columns are seen to give limited performance loss compared to the ML shape estimates, even though in addition also all correlation parameters were estimated. We conclude that the algorithm works well in high dimensions when the first PCs have large eigenvalues, so that these PC vectors can be robustly detected, avoiding the curse of dimensionality. Note that our example is stylized but realistic. In financial markets, we also typically observe a (few) dominant factor(s) in terms of variance. We have performed various simulation experiments with different distributions and different dimensionality, leading to similar results in terms of performance.

\begin{rem} 
In high-dimensions, it is possible to combine our algorithm with robust correlation estimators, such as eigenvalue clipping \citep{Bun2017} or shrinkage \citep{Ledoit2022}. In the first case, the eigenvalues of the lower PCs are replaced by a single eigenvalue such that the trace of the correlation matrix is preserved. This procedure lowers the number of parameters, because the correlation matrix is invariant for orthogonal transformations in the subspace of degenerate eigenvalues. Adjusting the eigenvalues with clipping or shrinkage is advised when the number of dimensions approaches the number of data points $d/n>0.1$. \citep{Ledoit2022} 
\end{rem}

\subsection{Case study}\label{subsecCase}

Next, we turn from simulated data to historic financial return data. We perform a case study of diversification and systemic risk in the global financial market. We are particularly interested in time scales of relevance to quantitative risk management and internal capital models, which are the main methods that financial institutions apply to protect themselves against systemic risk. We therefore use weekly return data for 20 major stock indices around the world. This leads to a lower dimension and to a lower amount of data compared to a copula study at an individual stock level with intradaily data, but consequently the study becomes more in line with the applications that we have in mind. An internal market risk capital model for banks is for example based on a 10-day horizon and for insurance companies on a one-year horizon \citep{McNeil2005}. The weekly data starts at 1 January 1998 and ends on 29 December 2022, resulting in 25 years of data, which is a typical data history for a financial institution when constructing internal models for market risk capital. Table \ref{tabIndices} shows the major global stock indices that are included in the analysis.

\begin{table}
\begin{center}
\caption{\label{tabIndices} Major global market stock indices that are included in the empirical analysis.}
\vspace{0.25cm}
\begin{tabular}{l l l l l l l} 
\hline
Nr.  & Index  &  Country    & & Nr.  & Index     &  Country   \\ 
\hline
 1 & DAX        &  Germany    & & 11 & Kospi     &  South Korea \\
 2 & FTSE 100   &  UK         & & 12 & Sensex    &  India       \\
 3 & CAC 40     &  France     & & 13 & IDX Comp. &  Indonesia \\
 4 & FTSE MIB   &  Italy      & & 14 & KLCI      &  Maleysia  \\
 5 & IBEX 35    &  Spain      & & 15 & SP 500    &  US  \\
 6 & AEX        &  Netherlands& & 16 & TSX Comp. &  Canada  \\
 7 & ATX        &  Austria    & & 17 & Bovespa   &  Brazil  \\
 8 & HSI        &  Hong Kong  & & 18 & IPC       &  Mexico  \\
 9 & Nikkei 225 &  Japan      & & 19 & All ord.  &  Australia  \\
 10 & STI       &  Singapore  & & 20 & TA 125    &  Israel  \\
\hline
\end{tabular}
\end{center}
\end{table}

\subsubsection{Time series model}\label{subsubsecTimeSeries}

We start the empirical analysis by applying the common AR(1)-GARCH(1,1) model to account for volatility clustering. The time series data of weekly logreturns $r_{i,t}$ for each index $i$ is therefore filtered as follows: 
\begin{eqnarray}
r_{i,t}&=& \delta_{G,i} + \phi_{G,i} r_{i,t-1}+ \epsilon_{i,t}, \quad 
\sigma^2_{G,i,t} = \omega_{G,i} + \alpha_{G,i} \epsilon^2_{i,t-1}+ \beta_{G,i} \sigma^2_{i,t-1}. \nonumber
\end{eqnarray}
We estimate the parameters $\delta_{G,i}, \phi_{G,i}, \omega_{G,i}, \alpha_{G,i}$ and $\beta_{G,i}$ using quasi-maximum likelihood, and we obtain the typical results for stock indices, namely small negative AR(1) coefficients for most stock indices, and highly significant GARCH(1,1) parameters indicative of volatility clustering.

A point of consideration for an internal capital model is stability, since enhanced temporal fluctuations in capital charges are not the preferred situation from a practical and regulatory point of view. A capital model that leads to a strong increase of capital requirements under market stress has the disadvantage that it requires financial institutions to raise additional capital or sell investments under unfavorable circumstances. When this is not desired, it is a consideration to replace the time-dependent volatility from a GARCH-model with either through-the-cycle or stressed estimates that are stable over time. For a similar reason, we do not consider time-varying copula parameters for the present application. Although a time-dependent copula can generate accurate short-term density forecasts by adapting to recent information in the market \citep{Opschoor2020, Oh2023}, such a copula model leads to lower stability over time. As a result, we focus in our study on the ability of static copula specifications to capture cross-sectional systemic risk. 

\subsubsection{Estimation results}\label{subsecCaseEst}

We proceed with the devolatized residuals $x_{i,t}=\epsilon_{i,t}/\sigma_{G,i,t}$, which we assume to be i.i.d. The pseudo-copula observations are obtained by using the rank-based transform with $u_{i,t} = \text{rank}(x_{i,t})/(n+1)$ and $n = 1303$ filtered return observations. We initialize the estimation procedure using the inverse normal transform $\Phi^{-1}(u_{i,t})$ and applying PCA. The obtained eigenvectors and eigenvalues are used as starting point for Algorithm \ref{algoPcc}, which we apply to obtain the full iterative solution of the GMM estimator for the hyperbolic-normal (HB-N) PCC. The procedure from Algorithm \ref{algoPcc} converges within a few iterations. The final vector weights are shown in Fig. \ref{figEigVecData}. The variable indices have the same order as in Table \ref{tabIndices}. 

The first PC-vector captures the movement of the market as a whole, since all vector weights $\hat{w}_{i,1}$ have the same sign. The Western European and the North American indices are found to have the strongest co-movement, while the Maleysian and the Indonesian indices have the weakest co-movement. The second PC gives the most important mode of diversification. In our case study, it describes the movement of the Asian indices ($i=8$-14) compared to the Western indices ($i=1$-7, 15-16). The other indices ($i=17$-20) have small vector weights $\hat{w}_{i,2}$, signalling that they fall geographically or economically in between these two categories. We also note that the Japanese index ($i=9$) has a smaller vector weight $\hat{w}_{9,2}$ than the other Asian indices in the second PC, meaning that it moves less strongly opposite to the Western indices. We conclude that detailed global diversification effects are automatically captured by the copula. The third PC mainly describes the joint movement of the American indices ($i=15$-18) relative to other indices, where especially the additional variance of the two Latin-American indices ($i=17,18$) is captured. 

\begin{figure}
\centering
\includegraphics[width=0.75\columnwidth]{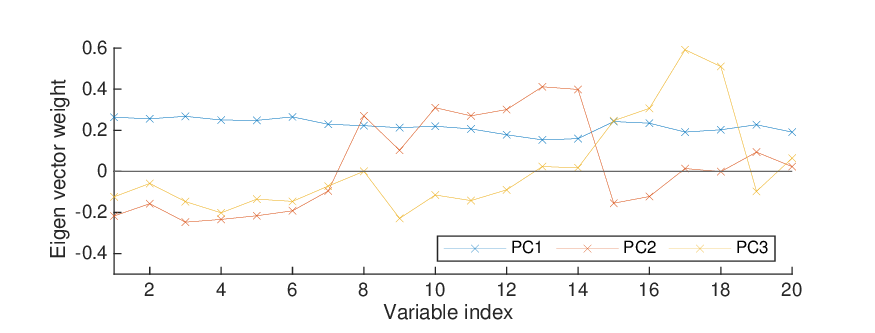}
\caption{Estimated eigenvectors weights $\hat{w}_{i,j}$ for the first three Principal Components $j \in \{1,2,3\}$ of the implicit copula as a function of the variable index $i\in\{1,...,20\}$. }\label{figEigVecData}
\end{figure}

The behavior of the estimated HB-N PCC is visualized using the steps of Fig. \ref{figSchemes}b. The result is shown in Fig. \ref{figHistSims}. Although the model is 20-dimensional, we select the FTSE-100 and the S\&P-500 indices for visualization at index level. The first plot shows historical filtered logreturns ($x_{i,t}$). The second plot shows historical (pseudo) copula observations ($u_{i,t}$). The estimated inverse marginal distribution functions $\hat{F}^{-1}_{Y_i}$ are used to obtain the third plot for the implicit copula returns ($y_{i,t}$). The PCA-based rotation is used to obtain the fourth plot for historical PC observations ($p_{j,t}$). When simulating, we perform the steps in the opposite direction. We start with sampling from the PC distributions to obtain $\tilde{p}_{j,s}$ and apply the (inverse) PCA-based rotation to obtain simulated implicit copula returns $\tilde{y}_{i,s}$. The marginal distribution functions $\hat{F}_{Y_i}$ are used to obtain simulated copula observations $\tilde{u}_{i,s}$ and the inverse empirical distribution functions $\hat{F}^{-1}_{X_i}$ are used to obtain filtered logreturns $\tilde{x}_{i,s}$. Fig. \ref{figHistSims} shows that we can visualize and assess the performance at each layer of the copula model. This holds in particular for the first two PCs, which are the core risk drivers that explain about 65\% of the total variance. Fig. \ref{figHistSims} also shows that the implicit copula returns $y_{i,t}$ have fewer outliers than the initial return observations $x_{i,t}$, which is beneficial for PCA. 

We perform the same estimation and simulation steps also for the Skew $t_1$-$t_{d-1}$ PCC from Eqs.~(\ref{eqNmvmix}) and (\ref{eqExPccSkewMvt}), leading to comparable visualisations of the copula. Finally, we introduce the Skew $t_1$-$t_1$ PCC, which is given by:
\begin{equation} \label{eqExPccSkewt}
P_1 \sim  {\rm Skew} \, t_1(\nu_1,\Lambda_1,\gamma_1), \quad  P_j \sim t_1(\nu_j,\Lambda_j) \, \forall  j \geq 2. 
\end{equation}
Here, ${\rm Skew}\, t_1(\nu_1,\Lambda_1,\gamma_1)$ and $t_1(\nu_j,\Lambda_j)$ denote one-dimensional distributions with degrees-of-freedom parameter $\nu_j$, asymmetry parameter $\gamma_1$, variance $\Lambda_j$ and zero mean. In case of the Skew $t_1$-$t_1$ PCC all generators are independent, where the first generator is modelled with the GH skew $t_1$ distribution and the higher generators with the $t_1$ distribution. This PCC therefore has similarities to the Skew $t_1$-$t_1$ factor copula studied by \cite{Oh2017}. The stochastic representation for this copula can be obtained along the lines of Eq. (\ref{eqNmvmix}) using an independent stochastic mixing variable $W_j$ for each of the higher generators $j$.

\begin{figure}[b]
\centering
\includegraphics[width=1.0\textwidth]{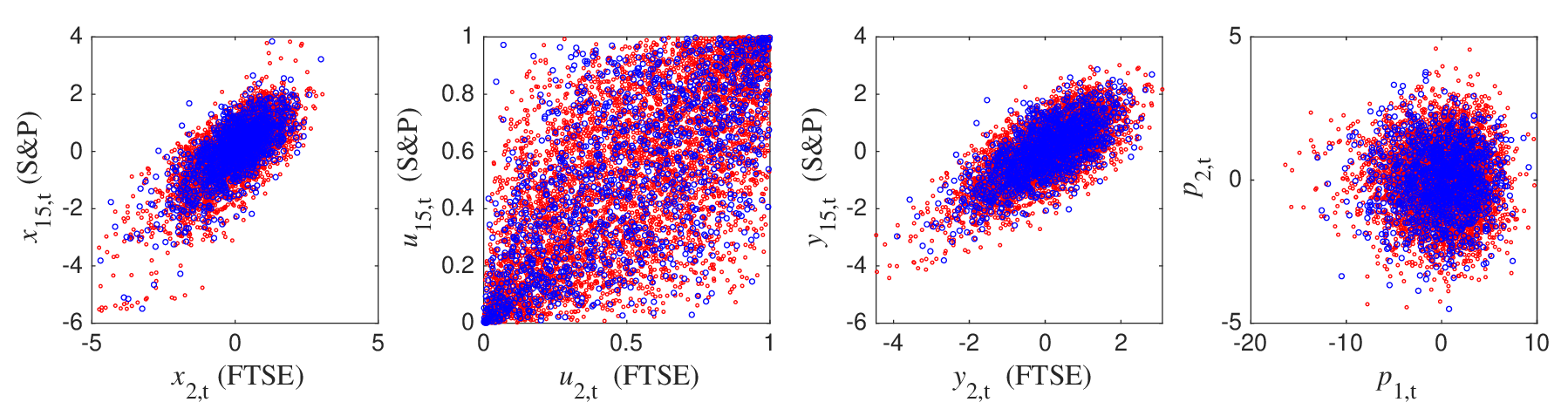}
\caption{Visualization of the model layers from Fig.\ref{figSchemes}b based on historical observations (blue dots) and on simulations from the estimated Hyperbolic-Normal PCC (red dots). Although the data and the model are 20-dimensional, the FTSE-100 and the S\&P-500 indices are used for visualization at index level. The left plot shows the (filtered) logreturn observations ($x_{i,t}$), the second plot shows the (pseudo-)copula observations ($u_{i,t}$), the third plot shows the implicit copula return observations ($y_{i,t}$) and the right plot shows the PC observations for the first two PCs ($p_{j,t}$). }\label{figHistSims}
\end{figure}

In Table \ref{tabEstCase20D}, the estimation results for various 20-dimensional copulas are shown. In all cases, we use a single DoF parameter $\nu$ in such a way that the generators have a common tail coefficient. We consider the Gaussian copula, the $t_d$ copula, and the GH skew $t_d$ copula as benchmark copulas. The Gaussian copula and the $t_d$ copula are most commonly used in practice at financial institutions. The Skew $t_d$ copula is an important extension due to its ability to capture asymmetry between up and down movements \citep{Yoshiba2018}. It has also been a popular choice in recent dynamic factor copula applications due to its analytic tractability \citep{Lucas2014, Ouyang2022, Oh2023}. We estimate the Gaussian, the $t_d$, and the Skew $t_d$ copula with ML estimation \citep{Yoshiba2018}. The Skew $t_d$ copula depends on a vector $\bm{\gamma}$ of asymmetry parameters which we align with the first eigenvector of the normal-mixture covariance matrix, so that $\bm{\gamma}=\gamma \bm{w}_{1}$. This choice aligns with the other PCCs, and the parallel direction $\bm{w}_{1}$ is also the relevant direction for including skewness in our case study. To determine the standard deviations of estimated parameters, we reperform all estimation procedures for 100 bootstrap samples of filtered returns \citep{Efron1994}. 

The results for the Gaussian copula are shown in Table \ref{tabEstCase20D} as a reference for the other copulas. As expected, the AIC and BIC improve significantly for the $t_d$ copula, as it incorporates tail dependence, and for the Skew $t_d$ copula, as it incorporates asymmetric dependence \citep{Lucas2014}. These normal mixture copulas are based on a single stochastic mixing variable. A novel aspect of our study is concerned with the optimal dimensionality of mixing variables. Recent factor copula studies have either assumed that all factors have one common mixing variable \citep{Opschoor2020, Ouyang2022, Oh2023}, or that all factors are independent \citep{Oh2017}. The HB-N and the Skew $t_1$-$t_{1}$ PCCs are also based on independent generators. \cite{Oh2017} have shown that such copula specifications lead to good performance for systemic risk, which we study in the next paragraph. However, these specifications are seen to be suboptimal in terms of likelihood. From Table \ref{tabEstCase20D}, we see that the Skew $t_1$-$t_{d-1}$ PCC leads to significant improvements in AIC and BIC over the other copula specifications. 

The interpretation for our results is that the Skew $t_1$-$t_{d-1}$ PCC benefits from having two stochastic mixing variables, namely $V$ for the first generator and $W$ for the higher generators in Eq. (\ref{eqNmvmix}). The first generator drives the parallel movement of all market indices, as shown in Fig. \ref{figEigVecData}, while the second and higher generators cause submarkets to move opposite to each other. The first generator is therefore important for systemic risk, while the second and higher generators lead to diversification. In crisis scenarios, a lack of diversification is expected, where all submarkets move in the same direction. This can be captured by a high-dimensional copula where $V$ and $W$ are decoupled, so that they do not grow large simultaneously. 

\begin{table}[t]
\begin{center}
\caption{\label{tabEstCase20D} Estimation results for the copula shape parameters based on 20-dimensional weekly return data. Also the loglikelihood $\ell$, and the improvement in AIC and BIC over the Gaussian copula are presented. Bootstrap standard deviations of estimates are shown in parentheses. The best performing copula, the Skew $t_1$-$t_{d-1}$ PCC, is highlighted in bold.}
\vspace{0.3cm}
\begin{tabular}{ l l l l l l l l} 
\hline
 & Gauss & $t_d$ & Skew $t_d$ & HB-N & Skew $t_1$-$t_1$ & Skew $t_1$-$t_{d-1}$ \\ 
\hline
$\hat{\nu}$    &      & 15.1 (1.2) & 14.0 (1.3)  &    & 12.7 (1.4)  &  14.9 (1.1)\\
$\hat{\gamma}$ &      &           & -2.7 (0.2) &       & -2.2 (0.5) & -2.8 (0.3) \\
$\hat{\alpha}$ &      &          &             & 1.3 (0.3) &       & \\
$\hat{\beta}$  &      &           &             & -0.7 (0.2) &       & \\
\hline
$ \ell$          &  10,505 & 10,895 & 11,001 & 10,668 & 10,812 & {\bf 11,041} \\
$ \Delta$ AIC    &   & -777 & -988 & -321 & -610 & {\bf -1068} \\
$ \Delta$ BIC    &   & -772 & -978 & -311 & -599 & {\bf -1058} \\
\hline
\end{tabular}
\end{center}
\end{table}

\subsubsection{High-dimensional dependence and systemic risk}\label{subsecCaseSystRisk}

In order to assess tail dependence and systemic risk, other measures than the likelihood are of importance. For pairwise dependence, the conditional probability of joint quantile exceedance (CPJQE) measures simultaneous exceedances of extreme quantiles. For downward exceedances, it is defined as: $\eta_{q,i,j}  = P[U_i \leq q \lvert U_j \leq q]$. The estimator for the CPJQE is obtained by:
\begin{equation}\label{eqCpjqe}
\hat{\eta}_{q,i,j}  =  
      \frac{1}{nq}\sum_{t=1}^n \bm{1}{ \{U_{i,t} \leq q, U_{j,t} \leq q \} }.
\end{equation}
This measure is often used by financial institutions to assess pairwise dependence over a range of quantiles $q$. The limit $q\rightarrow 0$ of $\eta_{q,i,j}$ leads to the coefficient of lower tail dependence $\eta_{i,j}$. In Fig. \ref{figCpjqe}, we show a pairwise CPJQE-plot for the FTSE-100 index and the S\&P-500 index as a function of the quantile $q$. We compare the empirical estimates $\hat{\eta}_{q,i,j}$ based on historic pseudo-copula observations with simulated estimates $\tilde{\eta}_{q,i,j}$ based on various copula specifications. For the estimated copula models, $10^6$ simulations are used to accurately determine $\tilde{\eta}_{q,i,j}$. The confidence intervals for the empirical estimates are determined by using $1,000$ bootstrap samples of filtered returns. The Gaussian copula and the Student $t_d$ copula are seen to have lower pairwise dependence in the tail than the historic data. The Skew $t_d$, the HB-N PCC, the Skew $t_1$-$t_{d-1}$ PCC and the Skew $t_1$-$t_1$ PCC lead to improved agreement. 

In order to assess systemic risk we are required to go beyond pairwise dependence. Instead of pairwise quantile exceedance, we consider the possibility of multiple market indices having a stressed downward return simultaneously. The fraction of market indices that have a joint downward exceedance of the quantile $q$ at time $t$ is given by the market distress ratio (MDR):
\begin{equation}  
{\rm MDR}_{t,q,d} =\frac{1}{d}\sum_{i=1}^d \bm{1} \{U_{i,t} \leq q \} .
\end{equation}
When the MDR approaches one, it means that all market indices are simultaneously in distress beyond the quantile $q$. The probability that the MDR exceeds a prespecified ratio $k/d$ was introduced by \cite{Oh2018} to study systemic risk in credit derivative markets. A similar measure was used by \cite{Ouyang2022} to study systemic risk in commodity markets. In both cases a dynamic copula was used to simulate the evolution of market distress probabilities over time.

\begin{figure}[t]
\centering
\includegraphics[width=0.85\columnwidth]{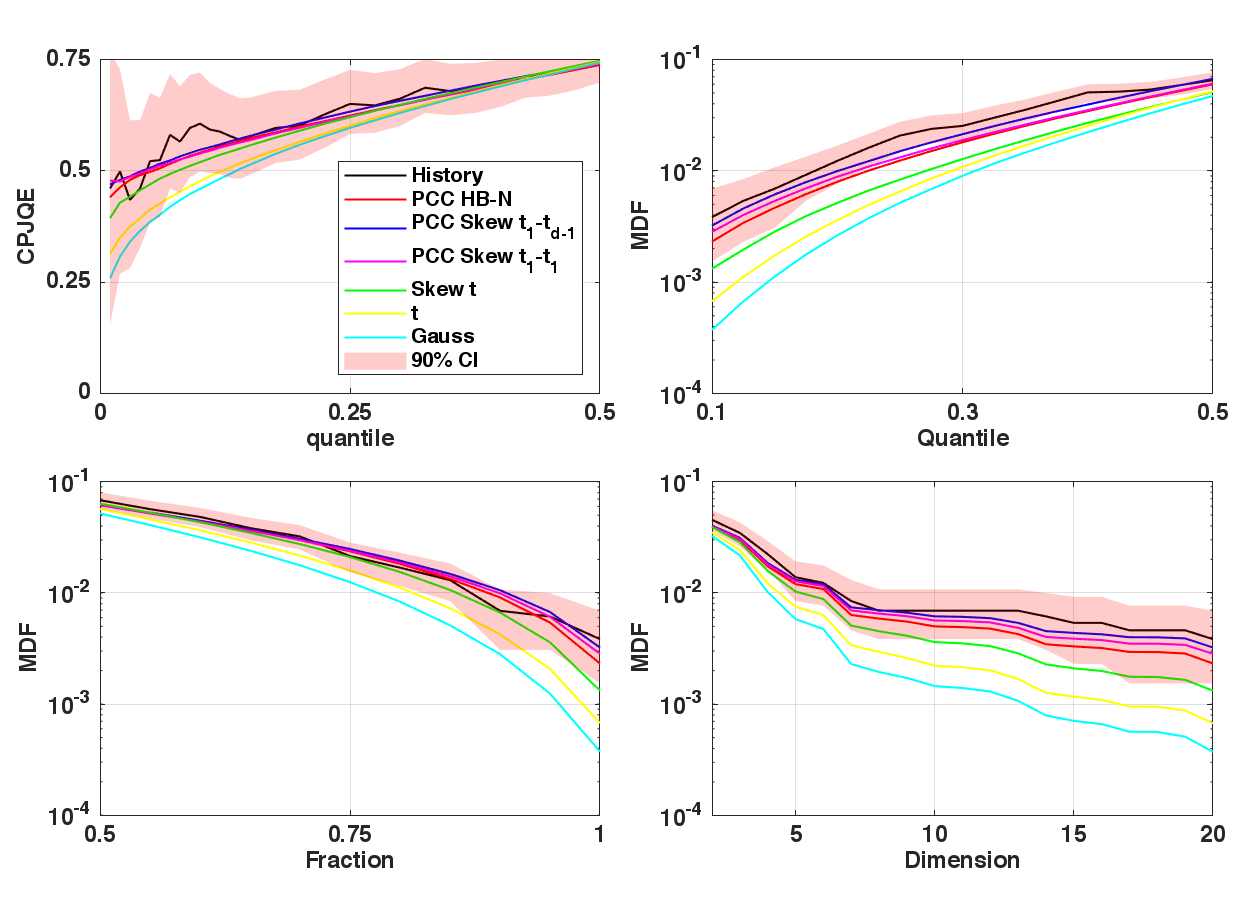}
\caption{Analysis of the market distress cube as a function of quantile, market fraction and market dimension. Upper left: pairwise CPJQE-plot $\hat{\eta}_{q,i,j}$ using historical pseudo-copula observations for the FTSE and S\&P index as a function of the quantile $q$ (black line). The CPJQE is also shown using simulations from various estimated (20-dimensional) copulae.  Upper right: Market distress frequencies ${\rm MDF}_{q,k,d}$ using historical pseudo-copula observations as a function of the quantile $q$ with $k/d=1$ and $d=20$ (black line). The MDF is also shown for the estimated copulae. Lower left: Market distress frequencies ${\rm MDF}_{q,k,d}$ as a function of the fraction $k/d$ with $q=0.1$ and $d=20$. Lower right: Market distress frequencies ${\rm MDF}_{q,k,d}$ as a function of the dimension $d$ with $q=0.1$ and $k/d=1$. In all plots bootstrap confidence intervals are given for the empirical estimates. The color coding is the same in all subplots. }\label{figCpjqe}
\end{figure}

In our study, we take a different approach. Namely, we perform a detailed comparison between observed market distress metrics and estimated copula specifications. We start by emphasizing that systemic risk can be studied as a function of an increasingly extreme quantile ($q$), as a function of an increasing market fraction in distress ($k/d$), and as a function of increasing market dimensionality ($d$). In all three directions the systemic risk grows, which brings us to the notion of the market distress cube. We perform a first empirical study of this cube, in order to assess the performance of various copula specifications in the three directions of systemic risk. In particular, we analyze the relative market distress frequency (MDF), which empirically measures the number of times that the market has been in distress:   
\begin{equation}  
{\rm MDF}_{q,k,d} = \frac{1}{n}\sum_{t=1}^n{\rm MDI}_{t,q,k,d} = \frac{1}{n}\sum_{t=1}^n \bm{1} \left\{{\rm MDR}_{t,q,d} \geq \frac{k}{d}\right\} 
\end{equation}
Here, the market distress indicator (MDI) specifies whether the stress event has occurred that $k$ or more market index returns simultaneously experienced a downward exceedance of the quantile $q$. The MDF is therefore an empirical counterpart of the joint distress probability considered by \cite{Oh2018}. 

Our analysis of the market distress cube is shown in Fig. \ref{figCpjqe}, where the relative market distress frequency ${\rm MDF}_{q,k,d}$ is studied in the three directions of systemic risk, namely the quantile ($q$), the market fraction ($k/d$), and the market dimensionality ($d$). In the right upper figure, the behavior of the MDF is shown as a function of $q$ on a logarithmic scale for $d=20$ and $k/d=1$, which means that all market indices simultaneously experience a downward exceedance of the quantile $q$. In the lower left plot, the behavior of the MDF is shown as a function of $k/d$ with $d=20$ and $q=0.1$. In the lower right plot, the behavior of the MDF is shown as a function of $d$ with $k/d=1$ and $q=0.1$. All plots of the MDF show that the estimated Gaussian copula, $t_d$ copula, and Skew $t_d$ copula lead to lower values of market distress frequencies than historically observed. The HB-N PCC, the Skew $t_1$-$t_1$ PCC and the Skew $t_1$-$t_{d-1}$ PCC perform better in capturing these measures of systemic risk. The Skew $t_1$-$t_{d-1}$ PCC shows the closest agreement with the historical data.

To further assess whether a copula specification is able to capture systemic risk, we propose to use a binomial test based on the Market Distress Indicator (${\rm MDI}_{t,q,k,d}$), which is Bernoulli distributed for a prespecified choice of $q$, $k$ and $d$. For each estimated copula, we accurately determine the model-implied distress probability based on a large number of simulations. After counting the number of historical MDI observations, we perform for each copula specification a one-sided binomial test with the null hypothisis of a correctly modelled distress probability against the alternative that the distress probability is larger than implied by the copula. The results for all copula specifications are shown in Table \ref{tabBinTest20D} for different combinations of the quantile $q$ and the fraction $k/d$ with $d=20$. As an example, there have been 16 observed weeks out of 1303 observations where all indices simultaneously experienced at least a 80\% worst-case loss ($q=0.2$). In case of the Gaussian copula and the $t$ copula, the $p$-value for this observation is lower than $10^{-4}$. For the Skew $t$ copula, the $p$-value is $0.002$. These normal mixture copulas would therefore be rejected by one-sided binomial tests at a significance level of 1\%. In case of the HB-N PCC copula, we find that the $p$-value of the test outcome is $0.06$, for the Skew $t_1$-$t_1$ PCC, the $p$-value is $0.12$ and for the Skew $t_1$-$t_{d-1}$ the $p$-value is $0.23$. As a result, the three PCCs pass the test for aggregate dependence at a significance level of 5\%.  

\begin{table}
\begin{center}
\caption{\label{tabBinTest20D} Binomial test results ($p$-values) for historical counts ($\#$) of the market distress indicator ${\rm MDI}_{t,q,k,d}$ based on 1303 filtered weekly observations of 20 major stock indices using three common high-dimensonal copulas and three PCCs. The test is performed for $d=20$, and various combinations of $q$ and $k/d$. The normal mixture copulas would be rejected, while the PCCs would be accepted at a significance level of 5\%.}
\vspace{0.3cm}
\begin{tabular}{ l l l l l l l l} 
\hline
$(q,k/d)$ & $\#$ & Gauss & $t_d$ & Skew $t_d$ & HB-N & Skew $t_1$-$t_1$ & Skew $t_1$-$t_{d-1}$ \\ 
\hline
(0.15,1)   & 9  & $<10^{-4}$ & 0.0006     & 0.014 & 0.16 & 0.27 & 0.41 \\
(0.2,1)   & 16 & $<10^{-4}$ & $<10^{-4}$ & 0.002 & 0.06 & 0.12 & 0.23\\
(0.15,0.9)  & 28 & $<10^{-4}$  & 0.0002  & 0.013 & 0.14 & 0.19 & 0.31 \\
(0.2,0.9)  & 44 & $<10^{-4}$   & 0.0001 & 0.009 & 0.11 & 0.13 & 0.26 \\
\hline
\end{tabular}
\end{center}
\end{table}

\section{Conclusion}
\label{sec:conc}

In this article, we have introduced the class of Principal Component Copulas, which combine the strong points of copula-based models with PCA-based models, such as detecting automatically the most relevant directions in high-dimensional data. Our study has given the first systematic analysis of the implicit copula structures that are generated by PCA. We have obtained novel theoretical results for PCCs. In particular, we have represented the full copula density in terms of characteristic functions. This allows computation of the density in high dimensions, while resolving only one-dimensional integrals that can be evaluated efficiently. We derived novel tail dependence properties of PCCs that are generated with exponential and Gaussian tailed distributions. This extends the existing analytic results on tail dependence for factor models which so far have relied on heavy tails. Furthermore, we have developed novel algorithms to estimate high-dimensional PCCs.
In order to achieve Maximum Likelihood Estimation, we have applied efficient numerical methods to compute the required one-dimensional integrals over characteristic functions such as the COS method. For an increasing number of parameters, we developed a novel algorithm that is designed to efficiently estimate a large number of correlation parameters using generalized method of moments and a small number of shape parameters based on MLE. We have shown excellent performance of this algorithm in a high-dimensional simulation experiment.

We have applied our techniques to the international stock market in order to study global diversification and systemic risk. We have shown that PCCs provide clear economic interpretations for their copula generators in terms of collective movements in the global market. Recent factor copula studies have either assumed that all factors have a common stochastic mixing variable, or that all factors are independent. We have proposed a PCC that is based on two stochastic mixing variables, namely one for the parallel direction and one for the orthogonal directions. This leads to enhanced performance on information criteria (AIC and BIC) and measures of systemic risk. We have expanded existing studies of systemic risk by introducing the notion of the market distress cube, in order to study systemic risk in three directions, namely more extreme crashes (quantile $q$), a larger market (dimension $d$), and increasing market share (fraction $k/d$). Finally, we have developed a binomial test to assess the ability of copulas to capture high-dimensional dependence. We find that PCCs pass such binomial tests, while commonly used copulas in finance fail to capture systemic risk, even when they perform well on pairwise dependence measures.  

Due to their analytic tractability, their explainability, and their ability to capture systemic risk, we consider PCCs to be promising for quantitative risk management and internal capital models. For further research, it would be interesting to allow for a dynamic dependence structure in the developed copula framework. This would make the copula less stable over time, but it would allow for more accurate short-term density forecasts based on recent information in the global market. Finally, because PCA and copula-based techniques are applicable to a wide range of data sets, it would be interesting to research the performance of the PCC on other data sets as well.  

\bibliography{references} 

\section*{Statements \& Declarations}

\subsection*{Funding}

The authors declare that no funds, grants, or other support were received during the preparation of this manuscript.

\subsection*{Competing Interests}

The views expressed in this article are the personal views of the authors and do not necessarily reflect the views or policies of their current or past employers. The authors have no competing interests.

\subsection*{Author contributions}

All authors contributed to the conception and design of the study. Data collection and technical implementation were performed by Koos Gubbels. The first draft of the manuscript was written by Koos Gubbels. All authors commented on previous versions of the manuscript. All authors read and approved the final manuscript.

\subsection*{Data availability}

The financial data for the case study is obtained from Yahoo Finance and is available upon request.

\subsection*{Code availability}

The code that is used to simulate and estimate Principal Component Copulas is available upon request.

\end{document}